\documentclass[a4paper, 11pt]{article}
\usepackage[OT1]{fontenc}
\usepackage{amsthm,amsmath}
\usepackage[numbers]{natbib}
\usepackage[colorlinks,citecolor=blue,urlcolor=blue]{hyperref}

\usepackage{amssymb, graphicx, latexsym}
\usepackage{arydshln} %dashline

\numberwithin{equation}{section}
\theoremstyle{plain}
\newtheorem{thm}{Theorem}[section]
\newtheorem{proposition}{Proposition}[section]
\newtheorem{lemma}{Lemma}[section]
\newtheorem{remark}{Remark}[section]

\usepackage{amsfonts}
\usepackage{authblk}

\usepackage[papersize={210mm,297mm},top=30mm, bottom=30mm, left=30mm , right=30mm]{geometry}

\begin{document}

%%%%%%%%%%%%%%%%%%% TITLE PAGE %%%%%%%%%%%%%%%%%%%

\title{The Limit Imbalanced Logistic Regression by Binary Predictors and its fast Lasso computation}

\author[1]{
Vincent Runge\footnote{E-mail: runge.vincent@gmail.com}}

% List of institutions
\affil[1]{LaMME - Laboratoire de Math\'ematiques et Mod\'elisation d'Evry.\newline UEVE - Universit\'e d'Evry-Val-d'Essonne.}
\date{}
\maketitle

% Abstract of the paper
\begin{abstract}
In this work, we introduce a modified (rescaled) likelihood for imbalanced logistic regression. This new approach makes easier the use of exponential priors and the computation of lasso regularization path. Precisely, we study a limiting behavior for which class imbalance is artificially increased by replication of the majority class observations. If some strong overlap conditions are satisfied, the maximum likelihood estimate converges towards a finite value close to the initial one (intercept excluded) as shown by simulations with binary predictors. This solution corresponds to the extremum of a strictly concave function that we refer to as "rescaled" likelihood. In this context, the use of exponential priors has a clear interpretation as a shift on the predictor means for the minority class. Thanks to the simple binary structure, some random designs give analytic path estimators for the lasso regularization problem. An effective approximate path algorithm by piecewise logarithmic functions based on matrix inversions is also presented. This work was motivated by its potential application to spontaneous reports databases in a pharmacovigilance context.
\end{abstract}

%\begin{keywords}

MS classification : Primary 62J12, 62F12, 62F15; secondary 34E05, 49M29, 62P10.

\section{Introduction}

If the response $y=1$ is very rare compared with the response $y=0$, we are in presence of a rare event configuration also called class imbalance. This problem recently got computer scientists' attention: they aimed at reducing computational costs by bypassing the class imbalance with resampling methods \cite{Guo} \cite{Oommen} \cite{Elrahman}. With these methods, the variance in estimating model parameters increases. Statisticians are aware of this problem and complex procedures such as local case-control sampling were proposed \cite{Fithian} (a method initiated in epidemiology \cite{Mantel}). 

In a recent work (2007) by Art B. Owen \cite{Owen}, the opposite approach is considered: the class imbalance is infinitely increased in order to reach the theoretical distribution of the majority class observations. Owen proved that under some overlap conditions the model parameters are finite (apart from the intercept) and built a limit system of equations related to exponential tilting, whose solution is the new estimate. The resulting equations include the  distribution of the infinite class expressed through integrals, which are not easy to infer. This may explain that this work was broadly ignored (The author found it when Sections \ref{section2} and \ref{section3} were already completed).  

In our approach, the observations of the majority class are infinitely replicated and the Owen's limit distribution becomes the observed distribution. This situation is a kind of degenerate case between resampling (we repeat observations) and infinitely class imbalance (the observed distribution is chosen as the theoretical one). Unlike Owen's result, our limit normal equations can be interpretated as the first order conditions of a new likelihood.\\

The idea of this work comes from the analysis of highly imbalanced binary spontaneous reports databases. Such databases are gathered by many countries and institutions (FDA, MHRA, WHO,...). Imbalanced logistic regression with binary predictors gives maximum likelihood estimate (MLE) very close to its limit imbalanced counterpart. This result makes possible the study of lasso-type regularization problem and the development of effective algorithms to provide model selection.

So far, only disproportionality methods are routinely used \cite{Madigan} for spontaneous report databases: predictors are analysed one by one, leading to a great number of false positive signals \cite{Harpaz}. Mathematical tools adjusted to binary data for regression are surprisingly barely developed by scientists (only boolean matrices have been studied by some authors \cite{Ki}). This results in an inflation of empiric methods using lasso regularization in recent years (from \cite{Caster} to \cite{Ahmed}). This is a worrying trend because recommendations made by these experts shift towards more complicated experimental methods and time-consuming algorithms, not towards a deeper mathematical understanding. This work is motivated by the need to better analyse this kind of applied problem. 

The paper contains three main sections in which we present the following results:\\

\begin{itemize}
\item 
In Section \ref{section2}, we investigate the properties of the logistic normal equations with binary predictors. Simple existence and uniqueness conditions of Silvapulle's type are found and some exact solutions presented. An invariance property in presence of intercept links this particular solution (called "square solution") to the limit imbalanced problem. We then acquaint ourselves with the issue of variance inflation of the imbalanced problem by computing the Fisher information.\\

\item In Section \ref{section3}, we derive Owen-type equations with a first order term evaluating the convergence rate. For the limit system of equations, the existence and uniqueness of the solution is proved with a new method leading to the minimization of a Kullback-Leibler divergence under linear constraints. A rescaling procedure on the initial likelihood and the previously found divergence justify the introduction of a rescaled likelihood corresponding to our limit imbalanced logistic regression problem. In a Bayesian framework, the Jeffreys penalty does not significantly decrease the variance of the estimator but other more appropriate priors, such that exponential ones, could help to reduce it (chosen according to the situation). The closeness in simulation between limit estimates and classical estimates compels us to go one step further with the study of regularization paths, in particular if the model is known to be sparse.\\

\item In Section \ref{section4}, we look at a lasso regularization problem for the rescaled likelihood, which has a clear interpretation as a shift on the predictor means for the class of interest. We succeed in finding some path estimators in a few particular cases (independence and orthogonal design). In presence of correlation, we present an effective path following algorithm by piecewise logarithmic functions giving precise estimates. We conclude by explaining the need of an analysis of the correlation structure between predictors. This leads to simple algorithmic procedures with small computational costs for which many different prior penalties could be easily tested. Two examples are given using the French spontaneous reports database.\\

\end{itemize}

The expressions "infinitely imbalance" and "limit imbalance" are considered as synonymous, although we recommend the use of the second one in our context due to the simple unique limit we impose and an analogy with hydrodynamic limits (in fluid dynamics) while the first expression is related to the underlying distribution introduced by Owen. \\

We conclude this article by discussing the many opportunities that arise with the introduction of a rescaled likelihood in a Bayesian context and of the path following algorithm by logarithmic functions.

\section{The logistic regression by binary predictors}
\label{section2}
\subsection{Logistic normal equations}

The binary logistic regression (BLR) problem consists in the determination of coefficients \(\hat{\beta}\) maximizing a smooth and concave likelihood function given by the relation
\[L(\beta|I_0,I_1,n^0,n^1) =  \prod_{i = 1}^{q_1}\left(\frac{e^{(I_1\beta)_i}}{1+e^{(I_1\beta)_i}}\right)^{n_i^1}\prod_{i = 1}^{q_0}\left(1+e^{(I_0\beta)_i}\right)^{-n_i^0}\,,\]
where \(\beta = (\beta_i) \in \mathbb{R}^{p+1}\) is indexed from zero with \(\beta_0\) corresponding to the intercept. Binary design matrices \(I_1 \in \mathcal{M}_{q_1 \times (p+1)}(\mathbb{B})\) and \(I_0 \in \mathcal{M}_{q_0 \times (p+1)}(\mathbb{B})\) with \(\mathbb{B} = \{0,1\}\) are of full rank: they aggregate the \(p\) binary predictors. Vectors of weights \(n^0=(n_1^0,...,n_{q_0}^0)^T \in (\mathbb{N}^*)^{q_0}\) and \(n^1=(n_1^1,...,n_{q_1}^1)^T \in (\mathbb{N}^*)^{q_1}\) save repetitions for distinct observations in response classes \(0\) and \(1\) separately. The binary structure favours repetitions in the sequence of observations, which justifies these notations. Moreover \((I_0\beta)_i\) is the i-th component of vector \(I_0\beta \in \mathbb{R}^{q_0}\) (the same for \((I_1\beta)_i\)). 

We introduce other notations thereafter used within this article. The modulus of a vector denotes its \(l^1\) norm, while the overline sign on lower cases stands for \(l^1\) normalization. For example \(|n^1| = \sum_{i=1}^{q_1}n_i^1\) and \(\overline{n}^1_i = n^1_i/|n^1|\) gives the vector \(\overline{n}^1\). $A_i$ is the i-th row of the matrix $A$ and its roman upper case equivalent $\mathtt{A}$ is the matrix $A$ in which the first column filled by ones (associated to the intercept) was removed. We also need \(N^1 = \mathtt{I}_1^Tn^1 \in \mathbb{R}^p\) with $T$ standing for the matrix transpose operator. An important feature in our study is the predictor means vector \(\overline{N}^1\) for class \(1\) obtained by the relation \(\mathtt{I}^T_1 \overline{n}^1 = \overline{N}^1\). For vectors of same size $u,\,v \in \mathbb{R}^q$, $uv$ (resp. $\frac{u}{v}$) is the vector with components $u_k v_k$ (resp. $\frac{u_k}{v_k}$), $k \in \{1,...,q\}$. $\tilde{\beta}$ is the vector $\beta$ without the intercept coefficient $\beta_0$. From Subsection \ref{subsection32}, the notations $I$ and $\mathtt{I}$ for matrices $I_0$ and $\mathtt{I}_0$ respectively are often used (as well as $q$ for integer $q_0$).\\

For ease of calculation, we consider the opposite of the log-likelihood. If \(I_1=I_0=\mathcal{I}\), we have \(q_1=q_0=q\) and we can introduce vectors \(n = n^1 + n^0\) and \(\Delta n = n^1 - n^0\). In this latter case, we write
\[l(\beta)=-\log(L(\beta)) =  |n|\log 2 + \sum_{i = 1}^{q}\left( -\Delta n_i(\frac{1}{2}(\mathcal{I}\beta)_i ) + n_i\log \cosh(\frac{1}{2}(\mathcal{I}\beta)_i)\right)\,,\]
and first order conditions are computed, differentiating $l$ with respect to each \(\beta_j\) coefficient. We obtain
\[0 = \frac{\partial l(\beta)}{\partial \beta_j} = \sum_{i = 1}^{q}\left( -\Delta n_i(\frac{1}{2}\mathcal{I}_{ij} ) + \frac{1}{2} n_i\mathcal{I}_{ij} \tanh(\frac{1}{2}(\mathcal{I}\beta)_i)\right)\,,\quad j \in \{0,...,p\}\,,\]
or in matrix form
\begin{equation}
\mathcal{I}^T \Delta n = \mathcal{I}^T \left(n \tanh(\frac{1}{2} \mathcal{I}\beta )\right)\,.
\label{normal}
\end{equation}
In a general framework with non-identical matrices \(I_0\) and \(I_1\), we likewise derive
\begin{equation}
I_1^Tn^1 - I_0^Tn^0 = I_1^T\left(n^1 \tanh(\frac{1}{2}I_1\beta)\right)+I_0^T\left(n^0 \tanh(\frac{1}{2}I_0\beta)\right)\,.
\label{normalgeneral}
\end{equation}

This system of equations (\ref{normalgeneral}) gathers the so-called logistic normal equations and will be widely used within this article.  

\begin{remark}
These equations are usually presented with a logistic function but we chose another expression to highlight the link with existence and uniqueness conditions.
\end{remark}

\subsection{Existence and uniqueness}

Necessary and sufficient conditions to ensure existence and uniqueness of the MLE are well-known, they were established by Silvapulle in 1981 \cite{Silva}. They consist in satisfying an overlap condition \(C_1\cap C_0 \ne \emptyset\) between the cones 
$$C_1 = \left\{ I_1^T u_1 \,|\, u_1 \in (\mathbb{R}^*_+)^{q_1}\right\}\,\, \hbox{and}\,\,C_0 = \left\{ I_0^T u_0 \,|\, u_0 \in (\mathbb{R}^*_+)^{q_0}\right\}\,.$$

For the BLR problem, a more convenient description is possible:

\begin{thm}
\label{Silvapulle}
The BLR problem admits a unique solution if and only if there exist $n^+ \in (\mathbb{N}^*)^{q_1}$ and $n^* \in (\mathbb{N}^*)^{q_0}$, such that $I_1^T n^+=I_0^T n^*$.
\end{thm}

Looking at equations (\ref{normalgeneral}), this theorem means that a MLE exists and is unique if one can find a couple $(n^+, n^*)$ of observations of the rows in $I_1$ and in $I_0$ such that $|n^+| = |n^*|$ vanishing all the regression coefficients (intercept included). An easy necessary condition to check is that at least one \(0\) and one \(1\) are present in each column of $I_0$ and $I_1$ (at the exception of the first column of ones corresponding to intercept).

\begin{proof}
If $I_1^T n^+ = I_0^T n^*$, the Silvapulle's condition is immediately verified. Reciprocally, $C_1 \cap C_0$ is an open subset of $\mathbb{R}^{p+1}$ with positive measure because $I_0$ and $I_1$ are full rank matrices. By a density argument, there exist $q \in (\mathbb{Q} \cap ]0,1[)^{p+1}$, $\lambda \in (\mathbb{R}^*_+)^{q_0}$ and $\mu \in (\mathbb{R}^*_+)^{q_1}$ satisfying $I^T_0 \lambda = I^T_1 \mu = q$. We reorder the rows in $I_0$ and $I_1$ such that the first $p+1$ rows are linearly independent. Let $H_0$ in $\mathcal{M}_{q_0 \times (q_0-p-1)}(\mathbb{R})$ and $H_1$ in $\mathcal{M}_{q_1 \times {(q_1-p-1)}}(\mathbb{R})$ be orthogonal matrices to $I_0$  and $I_1$ respectively. Because of the reorganization of the rows in $I_i$ ($i \in \{0,1\}$) we can choose a $H_i$ where its last $q_i-p-1$ rows form an identity matrix $\mathbb{I}_{q_i-p-1}$. For all $\alpha_0 \in \mathbb{R}^{q_0-p-1}$ and $\alpha_1 \in \mathbb{R}^{q_1-p-1}$ we have the relation $I^T_0(\lambda +H_0\alpha_0) = I^T_1 (\mu+H_1\alpha_1) = q$. Again with a density argument, we find $\alpha_0$ such that $\lambda_i + (\alpha_0)_i \in  \mathbb{Q}^*_+$ for all $i \in \{p+2,...,q_0\}$ and satisfying the constraint $\lambda +H_0\alpha_0 \in (\mathbb{R}_+^*)^{q_0}$.
For a matrix $A \in \mathcal{M}_{n \times m}(\mathbb{R})$, a vector $v \in \mathbb{R}^m$ and $J \subset \{1,...,m\}$, let $[Av]_J$ denote the vector $A_Jv_J$, where $A_J$ (resp. $v_J$) corresponds to the submatrix of $A$ (resp. subvector of $v$) obtained by removing from $A$ (resp. from $v$) the columns (resp. rows) that do not correspond to the indices in $J$. With this notation, we have $[I^T_0(\lambda +H_0\alpha_0)]_{\{1,...,p+1\}} = q - [I^T_0(\lambda +H_0\alpha_0)]_{\{p+2,...,q_0\}} \in \mathbb{Q}^{p+1}$. The binary matrix $(I^T_0)_{\{1,...,p+1\}}$ is then nonsingular and using its inverse in $\mathcal{M}_{(p+1)\times (p+1)}(\mathbb{Q})$ we obtain $(\lambda +H_0\alpha_0)_{\{1,...,p+1\}} \in (\mathbb{Q}_+^*)^{p+1}$. Finally $\alpha_0^* = \lambda + H_0 \alpha_0 \in (\mathbb{Q}_+^*)^{q_0}$. The same arguments lead to a set of coefficients $\alpha_1^* = \mu + H_1 \alpha_1 \in (\mathbb{Q}_+^*)^{q_1}$. Multiplying the vector $(\alpha_0^*,\alpha_1^*)$ by the ppcm of all its denominators proves the result.\end{proof}

\subsection{The square case}

The situation with identical square design matrices \(I_0\) and \(I_1\) is worthwhile in itself because it leads to explicit analytic formulae for the MLE and their variance (in the asymptotic case). In particular, we focus on the introduction of imbalance between $n^1$ and $n^0$ to emphasize the simple solution for MLE and the problem of variance inflation.

\begin{thm}

If $I_0 = I_1 = \mathrm{I}$ is a square matrix, we have the following closed form for the maximum likelihood estimator:
\begin{equation}
\label{solution}
\hat{\beta} = \mathrm{I}^{-1} \log \left(\frac{n^1}{n^0} \right)\,.
\end{equation}
\end{thm}

\begin{proof}
the matrix $\mathrm{I}$ verifies the condition $q=p+1$ and is nonsingular with $\mathrm{I}^{-1}$ its inverse (because $\mathrm{I}$ is of full rank). The vector $T$ is defined as \(T = \tanh \left(\frac{1}{2} \mathrm{I}\beta \right) \in \mathbb{R}^{p+1}\) i.e. \(\mathrm{I}\beta = \log\left(\frac{1+T}{1-T}\right)\). Multiplying (\ref{normal}) by \((\mathrm{I}^{-1})^T = (\mathrm{I}^T)^{-1}\), we get $T = \frac{\Delta n}{n}$. Hence, \(\beta = \mathrm{I}^{-1}\log\left(\frac{n+\Delta n}{n-\Delta n}\right)\), which achieves the proof.
\end{proof}

\begin{remark}
If one of the components in the vectors of weights \(n^1\) or \(n^0\) vanishes, some of the regression coefficients become infinite (but not necessarily all of them).\end{remark}

To our knowledge, this is the first general closed form found in the resolution of a logisitic regression. There exist partial results for a unique categorical predictor exposed by Lipovetsky in 2014 \cite{Lip}. An explanation for the lack of such a simple result stands in the poorly studied finite observation structure made possible through binary predictors with repetitions. In Appendix \ref{appA}, some particular solutions to equations (\ref{solution}) are presented. 

\subsubsection{Invariance if intercept}

We establish an invariance property making a link with the imbalanced problem.

\begin{proposition}
In the square case with intercept, multiplying all the components of $n^1$ or $n^0$ by a same integer does not change the value of the MLE apart from the intercept.
\end{proposition}

\begin{proof}
The inverse of a matrix with an intercept term verifies the relation $$\mathrm{I}^{-1} \begin{pmatrix}
   1 \\
   \vdots \\
   1 \\
   \end{pmatrix} =\begin{pmatrix}
   1 \\
   0\\
   \vdots \\
   0 \\
   \end{pmatrix}\,, $$
which means that we can rewrite equations (\ref{solution}) as
$$\hat{\beta}_i = \log\left(\prod_{j=0}^{p} \left(\frac{n^1_j}{n^0_j}\right)^{a_{ij}}  \right)\,,\, i \in \{0,...,p\}\,,\,\, a_{ij} = (\mathrm{I}^{-1})_{ij}\,\, \hbox{and}\,\,\sum_{j=0}^{p} a_{ij} = \delta_{0i}\,.$$
Substituting $n_j^0$ by  $s \times n_j^0$ (or $n_j^1$ by  $s \times n_j^1$) with $s \in \mathbb{N}^*$ gives the same result for all $\beta_i\,,\, i \in \{1,...,p\}$.
\end{proof}

\subsubsection{Asymptotic variance}

To conclude this section, we study the asymptotic behavior of the estimator for large \(|n^1|\) and \(|n^0|\). Since the MLE (intercept excluded) remains the same with or without a class imbalance (see the invariance property), we have a glimpse of a general property in class imbalance.

\begin{proposition}
\label{varsquare}
In the square case BLR problem, the variance of the maximum likelihood estimator is approximately given by relations
$$V(\hat{\beta}_i) \approx \sum_{j=0}^{p}a_{ij}^2\left(\frac{1}{n_j^1}+\frac{1}{n_j^0}\right)\,,\quad i \in \{0,...,p\}\,,\quad a_{ij} = (\mathrm{I}^{-1})_{ij}\,.$$
\end{proposition}
\begin{proof}
We compute the observed Fisher information $\mathcal{I}(\hat{\beta}) = \mathrm{I}^TD\mathrm{I}$ with $D$ a diagonal matrix with elements $n_i\hat{p}_i(1-\hat{p}_i)$ and $\hat{p}_i = 1/(1+e^{-(I\hat{\beta})_i})$. Its inverse gives the desired result, knowing that $n_i\hat{p}_i = n_i^1$ and  $n_i(1-\hat{p}_i) = n_i^0$.
\end{proof}

\begin{remark}
Another method uses the closed form (\ref{solution}) to perform variance and bias estimations by Taylor expansions with the multinomial random vector $(n^1,n^0)$. We obtain $V(\hat{\beta}_i) \approx \sum_{j=0}^{p}a_{ij}^2(\frac{1}{n_j^1}+\frac{1}{n_j^0}-\frac{2}{|n|})$ and $Bias(\hat{\beta}_i) \approx \sum_{j=0}^{p}\frac{a_{ij}}{2}(\frac{1}{n_j^0}-\frac{1}{n_j^1})\,, i \in \{0,...,p\}\,$. However, simulations give inaccurate results and only the Fisher information method should be retained. 
\end{remark}

We investigate the variation of the variance with respect to the sample size \(|n|\) and the value of the intercept \(\beta_0\) for a simple fixed model $(\beta_1,...,\beta_5) = (-0.5,-0.25,0,0.25,0.5)$. With these two parameters given, we simulate $10^4$ data sets with a different random binary square matrix \(\mathrm{I}\) and different random vectors \(n^1\) and \(n^0\) for each of them (but \(|n^1|+|n^0|\) is fixed). In table \ref{table1}, we compare the estimated standard deviation (\emph{sd.}) with the Fisher standard deviation given in Proposition \ref{varsquare} (\emph{F.sd.}) accompanied by an estimation of the bias (\emph{bias}) for coefficient \(\beta_4 = 0.25\).

\begin{table}[!ht]

\begin{center}
\begin{tabular}{c|ccccccccc}

&$\beta_0$ & \multicolumn{1}{c}{-7} & \multicolumn{1}{c}{-6} & \multicolumn{1}{c}{-5} & \multicolumn{1}{c}{-4} & \multicolumn{1}{c}{-3}& \multicolumn{1}{c}{-2}& \multicolumn{1}{c}{-1}& \multicolumn{1}{c}{0}\\
$|n|$&$|n_0|/|n_1|$ &1052& 385 & 142 & 52 & 19 & 7.1 & 2.7 & 1.0 \\
\hline
&sd.& . & . & . & . & 0.768 & 0.467 &  0.333 & 0.286  \\
$10^3$&F.sd.& . & . & . & . & 0.762 & 0.464 &0.334 & 0.293   \\
&bias& . & . & . & . & 0.019 & 0.0037 &0.0065 & -0.0021 \\
\hline
&sd.& . & . & 0.633 & 0.369 & 0.224 & 0.145 & 0.105 &  0.0931 \\
$10^4$&F.sd.& . & . & 0.618 & 0.361 & 0.221 & 0.145 & 0.104 & 0.0918   \\
&bias& . & . &  0.019 & 9.4e-3 & -2.4e-4 & 1.e-3 & -3.0e-4 & -2.9e-4 \\
\hline
&sd.&0.528 & 0.303 & 0.187 & 0.110 & 0.0690 & 0.0450 & 0.0328& 0.0300   \\
$10^5$&F.sd.& 0.512 & 0.301 & 0.183 & 0.111 & 0.0685 & 0.0451 &0.0327 & 0.0291   \\
&bias& 0.020 & 5.4e-4&3.4e-3 & -3.2e-5& -5.6e-5&4.5e-4&-2.8e-5 & 1.9e-4   \\
\hline
\end{tabular}
\caption{Variance analysis in the square case with intercept for coefficient \(\beta_4 = 0.25\). We used the following quantities : \((sd.)^2 = 10^{-4}\sum_{i=1}^{10^4}(0.25 - (\hat{\beta}_4)_i)^2\), \((F.sd.)^2 = 10^{-4}\sum_{i=1}^{10^4}V(\hat{\beta}_4)_i\) and \(bias = 10^{-4}\sum_{i=1}^{10^4}(0.25 - (\hat{\beta}_4)_i)\).}
\end{center}
\label{table1}
\end{table}

These simulations highlight the accuracy of the "Fisher variance" in all configurations, which is very close to the estimated one. Bias is negligible compared with variance. For a constant number of observations \(|n|\), the variance increases when the disbalance between classes strengthens. This variance inflation is a key issue in class imbalance, we further explain how one can easily add a prior information to a rescaled likelihood to deal with this problem (see Subsection \ref{subsection34}).

\section{Limit imbalanced study}
\label{section3}

\subsection{Owen-type equations}

The limit case consists in infinitely replicating the majority class observations as if the theoretical distribution of this class was the observed one. This is a degenerate case of the Owen's study, that is why we know that the intercept coefficient tends to minus infinity whereas other regression coefficients are finite if a stronger overlap condition is satisfied \cite{Owen}. For the limit equations, an information reduction for the majority class occurs: only the means of the predictors matter, the correlation structure in this class of interest "disappears".

The following proposition presents the logistic normal equations (\ref{normalgeneral}) in a new form with a remainder term arising in case of class imbalance.

\begin{proposition}
\label{linearimb}
For an imbalanced binary logisitic regression with a class size for response $y=0$ '$s$' times greater than the one for response $y=1$, we obtained the system of equations 
$$\frac{\bold{n_1^0}}{n_1^0}  = \overline{N}^1 + \frac{1}{s}\left(\frac{n_2^0}{(n_1^0)^2}\left(\frac{\bold{n_2^0}}{n_2^0}-\overline{N}^1 \right) - \frac{n_1^1}{n_1^0}\left(\frac{\bold{n_1^1}}{n_1^1}-\overline{N}^1 \right)\right)+ o(\frac{1}{s})\,,$$

with $ s = \frac{|n^0|}{|n^1|}   \gg 1$. We used notations:
$$n_1^0 = \sum_{i=1}^{q_0} \overline{n}^0_i e^{(\mathtt{I}_0 \tilde{\beta})_i}\,,\,\, n_1^1 = \sum_{i=1}^{q_1} \overline{n}^1_i e^{(\mathtt{I}_1 \tilde{\beta})_i} \,,\,\, n_2^0 = \sum_{i=1}^{q_0} \overline{n}^0_i e^{2(\mathtt{I}_0 \tilde{\beta})_i} \,,\,\, n_2^1 = \sum_{i=1}^{q_1} \overline{n}^1_i e^{2(\mathtt{I}_1 \tilde{\beta})_i}\,,$$
and for vectors in $\mathbb{R}^p$:
$$\bold{n_1^0} = \mathtt{I}_0^T (\overline{n}^0 e^{\mathtt{I}_0 \tilde{\beta}})\,,\,\, \bold{n_1^1} = \mathtt{I}_1^T (\overline{n}^1 e^{\mathtt{I}_1 \tilde{\beta}})\,,\,\, \bold{n_2^0} = \mathtt{I}_0^T (\overline{n}^0 e^{2\mathtt{I}_0 \tilde{\beta}})\,,\,\, \bold{n_2^1} = \mathtt{I}_1^T (\overline{n}^1 e^{2\mathtt{I}_1 \tilde{\beta}})\,.$$

\end{proposition}

The technical proof of this result is exposed in Appendix \ref{appB}.\\

As shown by simulations (see table \ref{table2}), the first order and remainder terms are negligible quantities with binary predictors, even if there is no imbalance! This suggests the introduction of the following limit imbalanced equations, obtained with $s = + \infty$ in Proposition \ref{linearimb}. 

\begin{thm}
\label{thmimba}
For infinitely imbalanced binary logisitic regression verifying a strong overlap condition (see Theorem \ref{surrounded}), the following system of $p$ limit imbalanced equations holds\footnote{With non-binary design matrices \(X_1\) and \(X_0\) and no vectors of weights, we obtain \(\mathtt{X}_0^T\left( \frac{e^{\mathtt{X}_0\tilde{\beta}}}{\sum_i e^{(\mathtt{X}_0\tilde{\beta})_i}}\right) = \overline{N}^1\,.\) These equations also differ from Owen's \cite{Owen}. }
\begin{equation}
\label{imbal}
\mathtt{I}_0^T\left( \frac{\overline{n}^0 e^{\mathtt{I}_0\tilde{\beta}}}{\sum_i \overline{n}^0_ie^{(\mathtt{I}_0\tilde{\beta})_i}}\right) = \overline{N}^1\,.
\end{equation}
Notice that the $\tilde{\beta}$ coefficients do not depend on the structure in rows of the design matrix associated to response $y=1$ but only on the means of ones for each predictor: $\overline{N}^1$. 
\end{thm}

We give a direct simple proof, avoiding the complicated previous proof of Appendix \ref{appB}.
\begin{proof} 
For $x$ near minus infinity, the hyperbolic tangent has the following first order expansion:
$$\tanh\left(\frac{x}{2} \right) = -1 + 2 e^{x} + o(e^{x})\,.$$
From \cite{Owen} we know that the intercept term tends to minus infinity, then with $x = I_0\beta$ or $x = I_1\beta$, we use the previous expansion neglecting the remainder term. Thus, equations (\ref{normalgeneral}) become
\begin{equation}
\label{gene}
I_1^Tn^1 = I_1^T\left(n^1 e^{I_1\beta}\right)+I_0^T\left(n^0 e^{I_0\beta}\right)\,,
\end{equation}
and factoring by $\exp(\beta_0)$ in the first equation of this system we have
\begin{equation}
\label{approxbeta0}
\exp(\beta_0)= \frac{|n^1|}{\sum_{i=1}^{q_1} n_i^1 e^{(\mathtt{I}_1\tilde{\beta})_i}+\sum_{i=1}^{q_0} n_i^0 e^{(\mathtt{I}_0\tilde{\beta})_i}} \approx \frac{|n^1|}{\sum_{i=1}^{q_0} n_i^0 e^{(\mathtt{I}_0\tilde{\beta})_i}}\,,
\end{equation}
because $\frac{|n^0|}{|n^1|} \to + \infty$. Looking back at (\ref{gene}) without the first equation, we have
$$\mathtt{I}_1^T \overline{n}^1 = \mathtt{I}_1^T\left(\frac{n^1 e^{\mathtt{I}_1\tilde{\beta}}}{\sum_{i=1}^{q_0} n_i^0 e^{(\mathtt{I}_0\tilde{\beta})_i}} \right)+\mathtt{I}_0^T\left(\frac{n^0 e^{\mathtt{I}_0\tilde{\beta}}}{\sum_{i=1}^{q_0} n_i^0 e^{(\mathtt{I}_0\tilde{\beta})_i}} \right)$$
but
$$\frac{n^1 e^{\mathtt{I}_1\tilde{\beta}}}{\sum_{i=1}^{q_0} n_i^0 e^{(\mathtt{I}_0\tilde{\beta})_i}} = \frac{|n^1|}{|n^0|} \frac{\overline{n}^1 e^{\mathtt{I}_1\tilde{\beta}}}{\sum_{i=1}^{q_0} \overline{n}_i^0 e^{(\mathtt{I}_0\tilde{\beta})_i}}\to 0$$
because $\frac{|n^1|}{|n^0|} \to 0$ and we obtain the desired result.
\end{proof}

In table \ref{table2}, we present simulation results based on limit imbalanced equations (\ref{imbal}) compared with classical logistic regression (\ref{normalgeneral}). The sample procedure is the same as the one used for table \ref{table1} except that we fixed sample size at $|n| = |n^1|+|n^0|= 10^4$ and vary dimension for the matrix $I_0$ (we chose \(q_0=10, 21, 32\)). 

\begin{table}[!ht]
\begin{center}

\begin{tabular}{c|ccccccccc}

&$\beta_0$ & \multicolumn{1}{c}{-5} & \multicolumn{1}{c}{-4} & \multicolumn{1}{c}{-3}& \multicolumn{1}{c}{-2}& \multicolumn{1}{c}{-1}& \multicolumn{1}{c}{0}\\
$q_0$&$|n^0|/|n^1|$& 141 & 51 & 19 & 7.0 & 2.6 & 1.0 \\
\hline
&sd.&  0.3614 & 0.2154 & 0.1331 & 0.08765 & 0.06384 & 0.05730  \\
&sd. imb.& 0.3614 & 0.2154 & 0.1331 & 0.08766 & 0.06402 & 0.05757 \\
\cdashline{2-8}
$10$&bias& 2.533e-3 & 7.572e-4 & -9.553e-4 & 1.407e-3 & 2.797e-4 & 4.954e-4   \\
&bias imb.& 2.532e-3 &7.527e-4 & -9.433e-4 & 1.439e-3 & 3.261e-4 & 5.073e-4 \\
\cdashline{2-8}
&$l^1$& 4.702e-4& 6.849e-4& 1.089e-3& 1.762e-3& 2.869e-3& 4.806e-3\\
\hline

&sd.&0.2643 & 0.1583 &  0.09879& 0.06468 & 0.04839& 0.04382 \\
&sd. imb.& 0.2643& 0.1583& 0.09879& 0.06474& 0.04860& 0.04451 \\
\cdashline{2-8}
$21$&bias& 2.148e-3& 1.956e-3& -1.583e-5& 7.601e-4& 2.461e-4&-6.280e-4  \\
&bias imb.& 2.130e-3& 1.967e-3& -8.450e-6& 7.752e-4& 3.000e-4& -5.337e-4 \\
\cdashline{2-8}
&$l^1$&5.950e-4& 8.864e-4& 1.399e-3& 2.281e-3& 3.698e-3& 5.952e-3 \\
\hline
&sd.& 0.2438& 0.1470& 0.09346& 0.06112& 0.04585&  0.04090 \\
&sd. imb.& 0.2438& 0.1471& 0.09349& 0.06117& 0.04628&  0.04171 \\
\cdashline{2-8}
$32$&bias&  4.385e-3& 1.034e-3& 8.025e-4& 3.584e-4& 1.259e-3& -1.905e-4   \\
&bias imb.& 4.383e-3& 1.048e-3& 7.894e-4& 3.537e-4& 1.313e-3& -4.297e-5 \\
\cdashline{2-8}
&$l^1$& 6.220e-4& 9.239e-4& 1.472e-3& 2.341e-3& 3.920e-3&  6.306e-03 \\
\hline

\end{tabular}
\caption{Variance and bias analysis in standard and imbalanced situations for coefficient \(\beta_4 = 0.25\). The \(l^1\) is given by the formula \(l^1 = 10^{-4}\sum_{i=1}^{10^4}|(\hat{\beta}_4^{imb})_i-(\hat{\beta}_4)_i|\).}

\end{center}
\label{table2}
\end{table}

The two estimates \(\hat{\beta}_4\) for standard and imbalanced regressions are very close to each other as shown by the mean of the $l^1$ norm -- even if the problem is not imbalanced -- so that standard deviation and bias are almost the same. This means that, if interesting properties can be established with the limit equations, this context will be appropriate to highlight new features in classical logistic regression. 

The $1/s$ first order term in Proposition \ref{linearimb} should be estimated to understand how good the limit imbalanced approximation is, without having to estimate the standard regression coefficients. Simulations show that this term is very small and we choose not to dwell on this intermediate situation, but it could be a more important result if non-binary design matrices are involved.

\subsection{Strong overlap condition and rescaled likelihood}
\label{subsection32}

Existence and uniqueness conditions to solve (\ref{imbal}) are well-known \cite{Owen}, they consist in an overlap condition a little bit stronger than the one given by Silvapulle. In fact, we need the point $\overline{N}^1$ to be surrounded by the rows of $\mathtt{I}_0$ (hereafter denoted by the letter \(\mathtt{I}\)). We give this result in the framework of the binary problem (simpler than Owen's general case) and establish a new proof leading to a minimum relative entropy problem. From there and using duality, we build the corresponding rescaled likelihood also justified by a rescaling on the initial likelihood.

\begin{thm}
\label{surrounded}
There exists a unique finite solution to the limit imbalanced BLR problem if and only if there exists $\lambda \in (\mathbb{R}_+^*)^{q}$ such that $\mathtt{I}^T \lambda = \overline{N}^1$ and $\sum_{i=1}^{q} \lambda_i = 1$. (If present, the null row (such that \(\mathtt{I}_i = (0,...,0)\)) is removed\footnote{in order to have non-zero coefficients \(\lambda\) as for the overlap condition in Theorem \ref{Silvapulle}.}.)
\end{thm}

\begin{remark} 
The condition $\mathtt{I}^T \lambda = \overline{N}^1$ means that we have \(I^T_0 \lambda = I^T_1 \overline{n}^1\) with \(\lambda \in (\mathbb{R}_+^*)^{q_0}\)  and \(\overline{n}^1 \in (\mathbb{R}_+^*)^{q_1}\) so that \(C_1\cap C_0 \ne \emptyset\). In other words, the existence and uniqueness of a solution for the limit problem implies existence and uniqueness  for its associated BLR problem.
\end{remark}

Our proof of this theorem is based on the following three lemmas. 

\begin{lemma}
\label{l1}
The log-sum-exp function  $h : \mathbb{R}^q \mapsto \mathbb{R}$, defined by $h(z)= \log(\sum_{i=1}^{q} e^{z_i})$ is a convex, continuous, increasing function on $\mathbb{R}^q$. The function $f : \mathbb{R}^p \mapsto \mathbb{R}$, $f(\tilde{\beta})= \log(\sum_{i=1}^{q} \overline{n}^0_i e^{(\mathtt{I}\tilde{\beta})_i})$ is continuous and convex on $\mathbb{R}^p$.
\end{lemma}

\begin{proof}
Function h has a positive semi-definite Hessian and is then convex. Furthermore for all $y,z \in \mathbb{R}^q$ such that $y_i \le z_i$, $i \in \{1,...,q\}$, we have $h(y) \le h(z)$ and the function is increasing on $\mathbb{R}^q$. The composition with an affine mapping preserves continuity and convexity. Thus, with $z = \mathtt{I}\tilde{\beta} + b$ and $\overline{n}^0 = e^{b}$ we obtain a convex continuous $f(\tilde{\beta}) = h(\mathtt{I}\tilde{\beta} + b)$ and $dom\, f = \mathbb{R}^p$. 
\end{proof}

\begin{lemma}
\label{l2}
The function $f : \mathbb{R}^p \mapsto \mathbb{R}$, $f(\tilde{\beta})= \log(\sum_{i=1}^{q} \overline{n}^0_i e^{(\mathtt{I}\tilde{\beta})_i})$ is strictly convex on $\mathbb{R}^p$.
\end{lemma}

\begin{proof}
The Hessian \(H\) of $h : \mathbb{R}^q \mapsto \mathbb{R}$, $h(z)= \log(\sum_{i=1}^{q} e^{z_i})$, is the following:
\[H_{ij} = \delta_{ij}\frac{e^{z_i}}{\sum_{k=1}^{q}e^{z_k}}-\frac{e^{z_i}}{\sum_{k=1}^{q}e^{z_k}}\frac{e^{z_j}}{\sum_{k=1}^{q}e^{z_k}}\,,\quad i,j \in \{1,...,q\}.\]
For all \(v=(v_1,...,v_q)^T \in \mathbb{R}^q\), we have
\[\sum_{i,j=1}^q v_iH_{ij}v_j = \frac{(\sum_{k=1}^{q}e^{z_k}v_k^2)(\sum_{k=1}^{q}e^{z_k})-(\sum_{k=1}^{q}e^{z_k}v_k)^2}{(\sum_{k=1}^{q}e^{z_k})^2}\,,\]
which is non-negative due to the Cauchy-Schwarz inequality. This expression is equal to zero if and only if there exists \(\lambda \in \mathbb{R}\) such that \(e^{z_k}v_k^2 = \lambda e^{z_k}\,,\forall k \in \{1,...,q\}\). Thus, only in the constant direction \(z_k(t) = t + z_k(0)\,, k \in \{1,...,q\}, t \in \mathbb{R}\), the function \(h\) is affine, in any others, this function is strictly convex.\\
Suppose that there exists a family of parameters \(F_a = \{\tilde{\beta}(t) \in \mathbb{R}^p, t \in [0,a], a >0\}\) such that \(z(t)=\mathtt{I}\tilde{\beta}(t) + b = t + z(0)\) and $e^{b} = \overline{n}^0$. This means that along the path described by \(\tilde{\beta}(t)\) the function \(f\) is affine. We obtain \(\mathtt{I}(\tilde{\beta}(t)-\tilde{\beta}(0))=t\) and with \(t \ne 0\), we have \(\gamma = (-t,\tilde{\beta}(t)-\tilde{\beta}(0))^T \in \mathbb{R}^{p+1} \setminus \{0\}^{p+1}\) such that \(I\gamma = 0\). This is impossible because the matrix \(I\) is of full rank, which proves the lemma.
\end{proof}

We present a corollary to a theorem on the Legendre-Fenchel transform of convex composite functions exposed in \cite{Hiriart}.
\begin{lemma}
\label{l3}
If functions $g_i : \mathbb{R}^p \mapsto \mathbb{R}$, $i \in \{1,...,q\}$ are convex and continuous with $dom\, g_i = \mathbb{R}^p$ and $h : \mathbb{R}^q \mapsto \mathbb{R}$ is convex, continuous and increasing with $dom\, h = \mathbb{R}^q$, then the convex conjugate of $h(g_1,...,g_q)$ is given by
$$[h(g_1,...,g_q)]^*(m) = \min_{\underset{m_1+...+m_q = m}{\alpha_1 \ge 0, ..., \alpha_q \ge 0}} \left(h^*(\alpha_1,...,\alpha_q) + \sum_{i=1}^{q}\alpha_i g_i^*(\frac{m_i}{\alpha_i}) \right)\,,$$
with $m \in (\mathbb{R}^p)^T$.
\end{lemma}

\begin{proof}[Proof of the theorem]
Let us define the function \(F_m\) such that
\[F_m : \left\{\begin{array}{cl}
			\mathbb{R}^p& \to \mathbb{R}\,,\\
			\tilde{\beta} &\mapsto m \cdot \tilde{\beta} - \log(\sum_i \overline{n}^0_i e^{(\mathtt{I}\tilde{\beta})_i})\,.
		\end{array}\right.\]	
\(F_m\) is differentiable on \(\mathbb{R}^p\) and the first order equations 
\[\frac{\partial}{\partial \tilde{\beta}_j}(\sum_{i=1}^{p}m_i \cdot \tilde{\beta}_i - f(\tilde{\beta})) = 0\,,\, j \in \{1,...,p\}\,,\]
are equal to the system (\ref{imbal}) with \(\overline{N}^1 = m^T\). Function \(F_m\) is strictly concave as the sum of a concave function and a strictly concave function (see Lemma \ref{l2}). Consequently, the solution \(\gamma\) to \(\nabla F_{\overline{N}^1}(\gamma) = 0\) is unique.

We now introduce the convex conjugate of the function $f$: 
\[f^* : \left\{\begin{array}{cl}
			(\mathbb{R}^p)^T &\mapsto \mathbb{R}\,,\\
			m  &\mapsto \sup\limits_{\tilde{\beta} \in \mathbb{R}^p}\left(m \cdot \tilde{\beta} - \log(\sum_i \overline{n}^0_i e^{(\mathtt{I}\tilde{\beta})_i})\right) = \sup\limits_{\tilde{\beta} \in \mathbb{R}^p}\left( F_m(\tilde{\beta})\right)\,.
		\end{array}\right.\]	
We will prove that the three following sets are identical
\[A = \bigg\{ m \in (\mathbb{R}^p)^T \,|\, \exists \tilde{\beta} \in \mathbb{R}^p\,,\, \mathtt{I}^T\left( \frac{\overline{n}^0 e^{\mathtt{I}\tilde{\beta}}}{\sum_i \overline{n}^0_ie^{(\mathtt{I}\tilde{\beta})_i}}\right) = m^T \bigg\}\,,\]

\[B = \bigg\{ m \in (\mathbb{R}^p)^T \,|\, f^*(m) < + \infty \bigg\}\,,\]

\[C = \bigg\{ m \in (\mathbb{R}^p)^T \,|\, \exists \, \lambda \in (\mathbb{R^*_+})^q\,,\, m = \lambda^T \mathtt{I} \,,\, \sum_{i=1}^{q}\lambda_i = 1 \bigg\}\,.\]

\underline{i) \(A \subset B\)}. If \(m_0 \in A\) there exists \(\tilde{\beta} \in \mathbb{R}^p\) solution to (\ref{imbal}), that is \(\nabla F_{m_0}(\tilde{\beta}) = 0\). Moreover \(f^*(m_0) = F_{m_0}(\tilde{\beta})\) because of the strict concavity of \(F_{m_0}\). Thus \(m_0 \in B\).

\underline{ii) \(B \subset C\)}. We use the Lemma \ref{l3} with $g_i(\tilde{\beta}) = (\mathtt{I}\tilde{\beta})_i +b_i$ and $h$ the log-sum-exp function verifying the necessary conditions (Lemma \ref{l1}). We have the convex conjugate $g_i^*(u_i)= -b_i$ if $u_i = \mathtt{I}_i$ and $+\infty$ elsewhere (we do not consider the presence of a null row \(\mathtt{I}_i=(0,...,0)\)). The only way to obtain a finite result is to impose the constraint $u_i = \frac{m_i}{\alpha_i} = \mathtt{I}_i$ for all $i \in \{1,...,q\}$. Therefore, knowing that
$$ h^*(\alpha_1,...,\alpha_q) = \left\{ \begin{array}{cc} \sum_{i=1}^{q}\alpha_i \log(\alpha_i) & if \,\, \alpha_1 \ge 0, ..., \alpha_q \ge 0\,,\,\alpha_1+...+\alpha_q = 1\,, \\ +\infty & otherwise\,,\end{array} \right.$$
we have
$$f^*(m)=[h(g_1,...,g_q)]^*(m) = \min_{\underset{\alpha_1\mathtt{I}_1+...+\alpha_q\mathtt{I}_q = m}{\underset{\alpha_1+...+\alpha_q = 1}{\alpha_1 \ge 0, ..., \alpha_q \ge 0}}} \left(\sum_{i=1}^{q}\alpha_i \log(\alpha_i) + \sum_{i=1}^{q}\alpha_i (-b_i) \right)$$
$$ = \min_{\underset{\alpha_1\mathtt{I}_1+...+\alpha_q\mathtt{I}_q = m}{\underset{\alpha_1+...+\alpha_q = 1}{\alpha_1 \ge 0, ..., \alpha_q \ge 0}}} \left(\sum_{i=1}^{q}\alpha_i \log(\frac{\alpha_i}{\overline{n}_i^0}) \right)\,.$$
We minimize a Kullback--Leibler divergence between two distributions under linear constraints. If one of the $\alpha_i$ is zero, $0g_i^*(\frac{m_i}{0}) = \sigma_{dom \, g_i}(m_i) = 0$ if $m_i=0$ elsewhere $ + \infty$ (see \cite{Hiriart}) and the previous equalities remain true with \(m_i = \alpha_i\mathtt{I}_i\). The KKT conditions of this problem impose the constraint \(\alpha_i >0\) for all \(i \in \{1,...,q\}\). Thus,
\begin{equation}
\label{KL}
f^*(m)=[h(g_1,...,g_q)]^*(m) =\min_{\underset{\alpha_1\mathtt{I}_1+...+\alpha_q\mathtt{I}_q = m}{\underset{\alpha_1+...+\alpha_q = 1}{\alpha_1 > 0, ..., \alpha_q > 0}}} \left(\sum_{i=1}^{q}\alpha_i \log(\frac{\alpha_i}{\overline{n}_i^0}) \right)\,.
\end{equation}
This minimum exists: this is a linear restriction to a convex and continuous function in a simplex and therefore \(B = dom \, f^* \subset C\).

\underline{iii) \(C \subset B \subset A\)}. If \(m_0 \in C\), then there exists \(\lambda \in (\mathbb{R}^*_+)^q\) such that \(\sum_{i=1}^{q}\lambda_i = 1\) and \(\lambda^T\mathtt{I} = m_0\) so that
\[f^*(m_0) = \sup_{\tilde{\beta} \in \mathbb{R}^p}\left(m_0 \cdot \tilde{\beta} - \log(\sum_i \overline{n}^0_i e^{(\mathtt{I}\tilde{\beta})_i})\right) = \sup_{\tilde{\beta} \in \mathbb{R}^p}\left(\log(\frac{e^{\sum_i \lambda_i(\mathtt{I}\tilde{\beta})_i}}{\sum_i \overline{n}^0_i e^{(\mathtt{I}\tilde{\beta})_i}}) \right)\,.\]
If the supremum is reached, there is a miximizing element \(\gamma \in \mathbb{R}^p\) and \(m_0 \in B\), this element is the solution to the system (\ref{imbal}) and thus \(m_0 \in A\). To state this result, it is enough to have \(-F_{m_0}\) coercive. Let \(\epsilon \in \mathbb{R}^p \setminus \{0\}^p\) be an arbitrary vector and \(\tilde{\beta} = x \epsilon\) with \(x \in \mathbb{R}\). Then,
\[F_{m_0}(x \epsilon) = \log \left(\frac{e^{\sum_i \lambda_i(\mathtt{I}\epsilon)_i x}}{\sum_i \overline{n}^0_i e^{(\mathtt{I} \epsilon)_i x}} \right) = \log \left( \frac{e^{\sum_i \lambda_i \omega_i x}}{\sum_i \overline{n}^0_i e^{\omega_i x}} \right)\,,  \]
with \(\omega = \mathtt{I}\epsilon\). Notice that the vector \(\omega\) can not satisfy the relations \(\omega_1 = ... =\omega_p\) because \(I\) is of full rank. Thus, if \(W = \max\limits_{i \in \{1,...,q\}}(\omega_i)\), we have
\[F_{m_0}(x \epsilon) = \log \left( \frac{e^{(\sum_i( \lambda_i \omega_i) - W) x}}{\sum_i \overline{n}^0_i e^{(\omega_i-W) x}} \right)\,,  \]
and \( \sum_i( \lambda_i \omega_i) - W<0\) because \(\sum \lambda_i = 1\). Therefore, with \(\Omega = \{i \in \{1,...,p\}\,|\, \omega_i = W\}\)
\[ \lim_{x\to +\infty} e^{(\sum_i( \lambda_i \omega_i) - W) x} = 0\,,\quad \lim_{x\to +\infty} \sum_i \overline{n}^0_i e^{(\omega_i-W) x} = \sum_{i \in \Omega}\overline{n}^0_i >0  \,,\]
which proves that the function \(-F_{m_0}\) is coercive when \(m_0 \in C\) and achieves the proof.
\end{proof}

The expression \(f^*(m)\) in (\ref{KL}) is the minimization of a relative entropy between the class 0 distribution and a kind of ghost class 1 distribution (built on the \(I = I_0\) design matrix). With the duality property, we can introduce a new likelihood. The following proposition leads to the same "limit" likelihood and justifies the use of the adjective "rescaled". Indeed:

\begin{proposition}
\label{rescaledlik}
The limit imbalanced equations arise from the following rescaled likelihood:
\[L^*(\tilde{\beta}|\mathtt{I}_0,\mathtt{I}_1,\overline{n}^0,n^1) = \prod_{j=1}^{q_1} \left( \frac{e^{(\mathtt{I}_1\tilde{\beta})_j}}{\sum_i \overline{n}^0_ie^{(\mathtt{I}_0\tilde{\beta})_i}} \right)^{n^1_j}\,.\]
\end{proposition}

\begin{proof}
With the initial likelihood
\[L(\beta|I_0,I_1,n^0,n^1) =  \prod_{j = 1}^{q_1}\left(\frac{e^{(I_1\beta)_j}}{1+e^{(I_1\beta)_j}}\right)^{n_j^1}\prod_{j = 1}^{q_0}\left(1+e^{(I_0\beta)_j}\right)^{-n_j^0}\,,\]
and the relation \(\exp(\beta_0) = \frac{|n^1|}{\sum_{i=1}^{q_0} n_i^0 e^{(\mathtt{I}_0\tilde{\beta})_i}+\sum_{i=1}^{q_1} n_i^1 e^{(\mathtt{I}_1\tilde{\beta})_i}}= \frac{|n^1|}{|n^0|} C(\frac{|n^1|}{|n^0|})\) (see \ref{approxbeta0}), we obtain the following expression for the likelihood, using notation \(x = \frac{|n^1|}{|n^0|}\):
\[x^{|n^1|}\prod_{j=1}^{q_1} \left(e^{(\mathtt{I}_1\tilde{\beta})_j}C(x)\right)^{n^1_j} \prod_{j=1}^{q_1} \left( \frac{1}{1 + x C(x)e^{(\mathtt{I}_1\tilde{\beta})_j}} \right)^{n^1_j} \prod_{j=1}^{q_0} \left( \frac{1}{1 + x C(x)e^{(\mathtt{I}_0\tilde{\beta})_j}} \right)^{n^0_j}\,. \]
We consider that \(|n^0|\) is large enough to consider the limit (with \(|n^1|\) fixed) \("x \to 0"\) and to make the approximations
\[\left(e^{(\mathtt{I}_1\tilde{\beta})_j}C(x)\right)^{n_j^1} \to \left(\frac{e^{(\mathtt{I}_1\tilde{\beta})_j}}{\sum_i \overline{n}^0_ie^{(\mathtt{I}_0\tilde{\beta})_i}}\right)^{n_j^1}\,,\quad \prod_{j=1}^{q_1} \left( \frac{1}{1 + x C(x)e^{(\mathtt{I}_1\tilde{\beta})_j}} \right)^{n^1_j} \to 1\,,\]
and
\[\prod_{j=1}^{q_0} \left( \frac{1}{1 + x C(x)e^{(\mathtt{I}_0\tilde{\beta})_j}} \right)^{n^0_j} \to e^{-|n^1|}\,, \]
thus, using a rescaling term,
\[L(\beta) \times \left(\frac{|n^0|}{|n^1|}\right)^{|n^1|}e^{|n^1|} \to L^*(\tilde{\beta})\,,\]
as \(\beta_0\) tends to minus infinity because \(\frac{|n^0|}{|n^1|} \to + \infty\). 
\end{proof}

The reader can see an analogy in physics with the existence of different scales of modelization. For example, the discrete mincroscopic N-body problem changed into the mesoscopic Boltzmann equation using the Boltzmann-Grad limit. See the book \cite{Saint} for further information on hydrodynamic limits.\\
This new likelihood makes now possible to consider a wide range of problems, related to variance reduction using simple prior penalties (Subsection \ref{subsection34}) or regularization (Section~\ref{section4}).

\subsection{The relative entropy dual problem}
With a likelihood and an entropy, we benefit from two points of view in order to numerically estimate the regression coefficients. The classical approach using a Newton-Raphson algorithm associated to the likelihood can be challenged by other algorithms on the primal or dual problems as described in \cite{Minka} and \cite{Yu} for classical logistic regression. We present here the dual problem and its link with initial regression coefficients. We leave the numerical analysis to another study.

\begin{proposition}
The regression coefficients of the limit imbalanced regression are given by the formulae
$$\hat{\tilde{\beta}} = (P^TP)^{-1}P^T \log\left( \frac{n^*}{\overline{n}^0} e^{-A}\right)\,,$$
where $n^*$ is the probability distribution solving a relative entropy problem with linear constraints
\[n^* = \operatornamewithlimits{argmin}\limits_{\underset{ \mathtt{I}^T \alpha = \overline{N}^1}{\underset{\alpha_1+...+\alpha_q = 1}{\alpha_1 > 0, ..., \alpha_q > 0}}} DL(\alpha || \overline{n}^0)\,,\]
$A = \sum_{i=1}^{q}n^*_i \log(\frac{n^*_i}{\overline{n}_i^0})$ and $P = \mathtt{I} - M$ with $P_{ij}=\mathtt{I}_{ij}-\overline{N}^1_j$.
\end{proposition}

\begin{proof}
With the existence of a unique solution (see Subsection \ref{subsection32}), there exists a solution \(n^* \in (\mathbb{R}^*_+)^q\) such that \(\mathtt{I}^T n^* = \overline{N}^1\), \(\sum_i n^*_i = 1\), and
$$\overline{N}^1 \cdot \tilde{\beta} - \log(\sum_i \overline{n}^0_i e^{(\mathtt{I}\tilde{\beta})_i}) = \sum_{i=1}^{q}n^*_i \log(\frac{n^*_i}{\overline{n}_i^0}) = A\,.$$
Then, using equations (\ref{imbal}) we obtain
$$\mathtt{I}^T n^* = \mathtt{I}^T (\overline{n}^0 e^{A}e^{P\tilde{\beta}})\,.$$
Let $H$ in $\mathcal{M}_{q \times (q-p-1)}(\mathbb{R})$ be an orthogonal matrix to $I$ (the previous relation remains true with \(I\) instead of \(\mathtt{I}\)) and $\gamma \in \mathbb{R}^{q-p-1}$, such that we can remove \(\mathtt{I}\) to obtain the relation
$$n^* + H \gamma = \overline{n}^0 e^{A}e^{P\tilde{\beta}}\,,$$
hence,
$$-\sum_{i=1}^{q}n^*_i \log(\frac{n^*_i}{\overline{n}_i^0})  + \log(\frac{n^*_k+(H\gamma)_k}{\overline{n}_k^0}) = \sum_{j=1}^{p}(\mathtt{I}_{kj}-\overline{N}^1_j)\tilde{\beta}_j\,,\quad k \in \{1,...,q\}\,.$$
Summing all these relations with weights \(n^*_k+(H\gamma)_k\), using the fact that \(\sum_{k=1}^{q}(H\gamma)_k=0\), gives
$$\sum_{i=1}^{q}n^*_i \log(\frac{n^*_i}{\overline{n}_i^0})  - \sum_{k=1}^{q}(n^*_k+(H\gamma)_k)\log(\frac{n^*_k+(H\gamma)_k}{\overline{n}_k^0}) = 0\,.$$
Due to convexity of the Kullback-Leibler divergence, we have a unique minimum obtained (by definition of \(n^*\)) at $\gamma =0$. Therefore

\[P \tilde{\beta} = \log \left(\frac{n^*}{\overline{n}^0}e^{-A} \right)\,,\]
and the result is proved if $P$ is of full rank. Suppose that this is not the case. Then, there exists $\gamma \in \mathbb{R}^p \setminus \{0\}^p$ such that \(P \gamma = (\mathtt{I}-M)\gamma = 0\), therefore \(\mathtt{I}\gamma = C\) with $C$ a vector with identical components all equal to $\sum_{i=1}^{p}\overline{N}^1_i \gamma_i$. Consequently, the matrix \(I\) (that is $\mathtt{I}$ with the intercept column of ones) is no more of full rank, which is, by definition of \(I = I_0\), impossible.
\end{proof}

\subsection{Priors for variance reduction and a priori information}
\label{subsection34}

The rare events structure of class imbalance goes hand in hand with the problem of precision for estimates. A classical solution consists in introducing an a priori distribution in a Bayesian context. This can be done using a Jeffreys non-informative prior \cite{Jeffreys} allowing both first order bias removal and variance shrinkage \cite{Firth}. Thus, we have to maximize the expression
\[L^*_J(\tilde{\beta}|\mathtt{I}_0,\mathtt{I}_1,\overline{n}^0,n^1) = \prod_{j=1}^{q_1} \left( \frac{e^{(\mathtt{I}_1\tilde{\beta})_j}}{\sum_i \overline{n}^0_ie^{(\mathtt{I}_0\tilde{\beta})_i}} \right)^{n^1_j} \times |\mathcal{I}(\tilde{\beta})|^{1/2}\,,\]
with \(|\mathcal{I}|\) the determinant of the Fisher information matrix. This approach is implemented in the R package \emph{logistf} for logistic regressions. In the imbalanced case, we search for a method conserving the shape of the limit equations and achieving at the same time variance reduction: we choose the following approximation
\[\frac{1}{2}\log(|\mathcal{I}(\tilde{\beta})|) \approx \frac{1}{2}\sum_{i=1}^{p} \log\left( \frac{\mathtt{I}_i^T(\overline{n}^0 e^{\mathtt{I}\tilde{\beta}})}{\sum_j \overline{n}^0_je^{(\mathtt{I}\tilde{\beta})_j}}-(\frac{\mathtt{I}_i^T(\overline{n}^0 e^{\mathtt{I}\tilde{\beta}})}{\sum_j \overline{n}^0_je^{(\mathtt{I}\tilde{\beta})_j}})^2\right)\,,\]
supposing an absence of correlation between predictors in a random design framework (see Section \ref{section4}). With this hypothesis, we derive first order equations
\[N^1_i - |n^1|\frac{\mathtt{I}_i^T\overline{n}^0 e^{\mathtt{I}\tilde{\beta}}}{\sum_j \overline{n}^0_je^{(\mathtt{I}\tilde{\beta})_j}} + \frac{1}{2}\left(1 - 2\frac{\mathtt{I}_i^T(\overline{n}^0 e^{\mathtt{I}\tilde{\beta}})}{\sum_j \overline{n}^0_je^{(\mathtt{I}\tilde{\beta})_j}} \right) = 0\,,\quad i \in \{1,...,p\}\,,\]
thus,
\[\mathtt{I}^T\left(\frac{\overline{n}^0 e^{\mathtt{I}\tilde{\beta}}}{\sum_j \overline{n}^0_je^{(\mathtt{I}\tilde{\beta})_j}}\right) = \frac{N^1 + \frac{1}{2}}{|n^1|+1} = \overline{N}^1_J\,.\]
In table \ref{table3}, we simulate data sets as previously done with the length for \(n^0 = (n^0_1,...,n^0_{10})^T\) fixed (to \(10\)) and we compare estimated bias and variance for coefficient \(\beta_4 = 0.25\) with three different methods: a classical logistic regression (\emph{bias} and \emph{sd.}), the imbalanced case with means \(\overline{N}^1_J\) (\emph{im. bias} and \emph{im. sd.}) and the Jeffreys exact penalty (\emph{J. bias} and \emph{J. sd.}).

\begin{table}[!ht]

\begin{center}
\begin{tabular}{cccccccc}

$\beta_0$ & \multicolumn{1}{c}{-5} & \multicolumn{1}{c}{-4} & \multicolumn{1}{c}{-3}& \multicolumn{1}{c}{-2}& \multicolumn{1}{c}{-1}& \multicolumn{1}{c}{0}\\
$|n^0|/|n^1|$& 143 & 52 & 19 & 7.1 & 2.7 & 1.0 \\
\hline
sd.&  0.3720& 0.2173 & 0.1328 & 0.08835& 0.06388& 0.05645  \\

im. sd.& 0.3672&  0.2165 & 0.1326& 0.08837 & 0.06406 & 0.05677\\
J. sd.& 0.3575 & 0.2141 & 0.1322& 0.08814 & 0.06382 & 0.05641\\
\cdashline{1-7}
bias& 6.367e-3& 4.247e-3 & -8.032e-4 &  1.281e-3 & -2.822e-4 & 3.048e-4   \\
im. bias& 3.284e-3 &  3.025e-3& -1.176e-3& 1.159e-3& -3.229e-4& 2.987e-4 \\
J. bias& -1.631e-3&  1.419e-3& -1.858e-3& 8.354e-4& -5.028e-4& 1.345e-4 \\
\hline
\end{tabular}
\caption{Variance and bias analysis with prior distributions for coefficient \(\beta_4 = 0.25\)}
\end{center}
\label{table3}
\end{table}

Variance reduction is about 2 percents with the Jeffreys prior and the half as much its easily computable approximation in class imbalance. Bias was already small and gets a little smaller. The shrinkage of the variance is limitated by the Cram\'er-Rao bound (see Fisher variance in table \ref{table1}) and no miraculous reduction was conceivable.\\

In the next section, we consider path following methods to complete regularization and highlight its "simplicity" with binary data. The initial parameters being the maximum a posteriori estimate (MAP), this estimation is a central problem of the limit imbalanced study. The benefit of the rescaled likelihood compared with the standard one is in the easy use of exponential a priori penalties. Indeed, with the penalty\footnote{P could be written as a probability distribution with a normalization term (the support of regression coefficients is finite).}
\begin{equation}
\label{apriori}
P(\tilde{\beta}) = \exp\left(\sum_{i=1}^{p} \epsilon_i \tilde{\beta}_i\right)\,,
\end{equation}
where \(\epsilon \in \mathbb{R}^{p}\), we maintain the shape of the likelihood by only perturbing the predictor means vector \(\overline{N}^1\) by \(\frac{\epsilon}{|n^1|}\) (the MAP exists if and only if \(\overline{N}^1+\frac{\epsilon}{|n^1|}\) is surrounded by the rows of \(I\), see Theorem \ref{surrounded}).

\section{Path estimators for Lasso-type regularization}
\label{section4}

In ths section, we consider that each observation \(\mathtt{I}_i\) (\(i \in \{1,...,q\}\)) is generated by a random binary vector \(X_i^T = (X_{i1},...,X_{ip})^T \in \{0,1\}^p\) with \(\mathbb{E}[X_{ij}]=b_j\in \,]0,1[\), \(j \in \{1,...,p\}\). With this modelization, we find many path estimators depending on the underlying correlation structure of the random design. 

\subsection{Limit lasso properties}
The well-known lasso regularization consists in introducing a positive parameter $\lambda$ defining the strength of a Laplace prior distribution \cite{Tibshirani}. We search for the maximum of the expression
\[{\cal L} (\beta, \lambda) = L^*(\beta)\times \exp\left( - \lambda \sum_{i=1}^{p} |\beta_i|\right)\,,\]
which verifies the following simple first order conditions. Notice that we use, from now on, the notation $\beta$ instead of $\tilde{\beta}$ to facilitate the reading.

\begin{proposition}
The limit imbalanced BLR problem with lasso penalty leads to the system of equations
\begin{equation}
\mathtt{I}^T\left( \frac{\overline{n}^0 e^{\mathtt{I}\beta}}{\sum_i \overline{n}^0_ie^{(\mathtt{I}\beta)_i}}\right) = \overline{N}^1 - t \nu(\beta) \,,
\label{LassoIm}
\end{equation}
with \( t = \frac{\lambda}{|n^1|}\) and \(\nu_j(\beta) = {\rm sign}(\beta_j)\), if \(\beta_j \ne 0\), \(\nu_j(\beta) \in [-1,1]\), if \(\beta_j = 0\) for all \(j \in \{1,...,p\}\) (\(\nu\) is the subgradient of the \(l^1\) norm). 
\end{proposition}

Thus, the lasso has a clear interpretation as a shift operating on the observed proportions $\overline{N}^1$. Thereafter, we often use the vector \(p(t) \in \mathbb{R}^p\) defined as \(p(t) = \overline{N}^1 - t \,\nu(\beta)\). 

\begin{proposition}
\label{hatbetacont}
If the strong overlap condition in Theorem \ref{surrounded} is satisfied, then the function
\[\hat{\beta} : \left\{\begin{array}{cl}
			\mathbb{R}& \to \mathbb{R}^p\,,\\
			t &\mapsto \operatornamewithlimits{argmax}\limits_{\beta \in \mathbb{R}^p}({\cal L} (\beta,t))\,,\end{array}\right.\]

is continuous for all \(t \ge 0\) and there exists \(T \in [0,1]\), such that \(\beta(t)=0\,,\, \forall t \ge T\).
\end{proposition}

\begin{proof}
With the positivity of \(\cal L\), we have, 
\[\operatornamewithlimits{argmax}\limits_{\beta \in \mathbb{R}^p}({\cal L} (\beta,t)) = \operatornamewithlimits{argmin}\limits_{\beta \in \mathbb{R}^p}(-\log ({\cal L} (\beta,t)))\,,\]
and for all \(t \ge 0\), \(-\log ({\cal L})\) is a strictly convex and coercive function in \(\beta\) if the strong overlap condition is satisfied (see proof of Theorem \ref{surrounded}). Therefore, the function \(\hat{\beta}\) is well defined for all \(t \ge 0\). Furthermore, this function is continuous because of the continuity in \((\beta,t)\) of \(\log (\cal L)\) and its strict concavity in \(\beta\). The equations (\ref{LassoIm}) with \(t=1\) have no solution if one of the components of \(\nu(\beta)\) is equal to \(-1\) or \(+1\), therefore \(\hat{\beta}_j(1) = 0\) for all \(j \in \{1,...,p\}\).
\end{proof}

\begin{remark}
Using the law of large numbers, the family of model parameters \(\{\beta(t)\}_{t \ge 0}\) solves the system of equations

\begin{equation}
\label{betaesti}
\frac{\mathbb{E}[X_{1j}e^{X_1\beta(t)}]}{\mathbb{E}[e^{X_1\beta(t)}]} = p_j(t)\,,\quad j \in \{1,...,p\}\,,
\end{equation}
with \(\mathbb{E}\) being the expectation operator. This previous system of equations takes the same form as in (\ref{LassoIm}) because \(X_1\) is a discrete random vector and therefore the path estimator \(\{\beta(t)\}_{t \ge 0}\) is continuous. Notice that the function \(\nu \circ \beta\): \(\mathbb{R}^+ \to \mathbb{R}\) is also continuous in \(t\).
\end{remark}

\subsection{Path estimators}
Thanks to this previous remark, we are able to find precise analytic estimators of the path in the case of independent and orthogonal random designs. Notice that such solutions already exist in the framework of linear regression (see \cite{Tibshirani}). From now on, the strong overlap condition is considered to be always satisfied at \(t=0\).

\begin{thm}
\label{descriptionLasso}
If the random vector \(X\) generating the observations \(\mathtt{I}\) has independent components, a precise path estimator $\{\hat{\beta}(t)\}_{t \ge 0}$ is given by the formulae
$$\hat{\beta}_j(t) = \hat{\beta}_j(0) + \log\left(\frac{1-t\frac{{\rm sign}(\hat{\beta}_j)}{\overline{N}_j^1}}{1+t\frac{{\rm sign}(\hat{\beta}_j)}{1-\overline{N}_j^1}}\right)\,,\quad j \in \{1,...,p\}\,,$$
$${\rm if} \quad t \in \left[ 0,\, t_{0j}\right]\,,\quad t_{0j} =  \overline{N}_j^1\frac{|1-e^{-\hat{\beta}_j(0)}|}{1+\frac{\overline{N}_j^1}{1-\overline{N}_j^1}e^{-\hat{\beta}_j(0)}}\quad {\rm and} \quad \hat{\beta}_j(t) = 0\,, \quad {\rm if} \,\,t > t_{0j}\,.$$
The coefficients \(\hat{\beta}(0)\) are give by the classical MLE (solution of equations (\ref{normalgeneral}) without intercept) if we want to estimate the path obtained by an (imbalanced) logisitic regression. If we use the limit equations, we need the MLE of the rescaled likelihood (Proposition \ref{rescaledlik} and equations (\ref{thmimba})) and in this case:
\[\hat{\beta}_j(0) = \log \left( \frac{\overline{N}^1_j}{1-\overline{N}^1_j} \frac{1-\overline{N}^0_j}{\overline{N}^0_j}\right)\,,\quad j \in \{1,...,p\}\,.\]
\end{thm}

\begin{proof}

For all \(j \in \{1,...,p\}\) we use the hypothesis of independence:
\[p_j(t) = \frac{\mathbb{E}[X_je^{X\beta(t)}]}{\mathbb{E}[e^{X\beta(t)}]} = \frac{\mathbb{E}[X_je^{X_j\beta_j(t)}]\mathbb{E}[\prod_{k \ne j}e^{X_k\beta_k(t)}]}{\mathbb{E}[e^{X_j\beta_j(t)}]\mathbb{E}[\prod_{k \ne j}e^{X_k\beta_k(t)}]}= \frac{\mathbb{E}[X_je^{X_j\beta_j(t)}]}{\mathbb{E}[e^{X_j\beta_j(t)}]}\]
\[ = \frac{e^{\beta_j(t)}P(X_j = 1)}{e^{\beta_j(t)}P(X_j = 1) + P(X_j = 0)} = \frac{e^{\beta_j(t)} b_j}{e^{\beta_j(t)}b_j + (1-b_j)}\,,\]
and the solution is
\[\beta_j(t)= \log \left(\frac{p_j(t)}{1-p_j(t)}\right) - \log \left( \frac{b_j}{1-b_j}\right)\,, \quad t \in \left[0, |\overline{N}_j^1-b_j|\right]\,,\]
and \(\beta_j(t) = 0\) if \(t > |\overline{N}_j^1-b_j|\). Indeed, $\dot{\beta_j}$ is negative in region \(\beta_j >0\) and positive in region \(\beta_j <0\). \(\beta_j(0) = \log\left(\frac{\overline{N}_j^1}{1-\overline{N}_j^1} \frac{1-b_j}{b_j} \right)\) for a random design with independent predictors (see Appendix \ref{A2} with \(p=1\)). We replace all the \(b_j\) by the frequencies of observations \(\overline{N}^0_j\) to obtain the estimator.
\end{proof}

The orthogonal case, when the inner product between columns of the design matrix vanishes (\(X_{1j}X_{1k}=0\), \(j \ne k\)), is also tractable.

\begin{thm} 
\label{thindep}
If the random design is orthogonal, we have \(\mathtt{I} \in \mathcal{M}_{(p+1) \times p}(\mathbb{B})\) filled by zeros except at positions \((i+1,i)\), \(i = 1,...,p\) and the derivative of the path estimator takes the form
$$\dot{\hat{\beta}}_i(t) = \dot{\overbrace{\log \left( \frac{p_i(t)}{1-\sum_{s \in S_t}p_s(t)} \right)}}\,,\quad i \in S_t\,,\quad t \ge 0\,,$$
with \(\{S_t\}_{t \ge 0}\) a family of subsets of \(\{1,...,p\}\) containing the indexes of non-zero coefficients of vector \(\beta\) at time \(t \ge 0\). The algorithm that describes the positions of the change-points in \(S_t\) is described in the proof.
\end{thm}

\begin{proof}
With the hypothesis of orthogonality, equations (\ref{betaesti}) are reducted to
\begin{equation}
\left\{\begin{array}{cl}
\frac{b_1e^{\beta_1(t)}}{b_0 + b_1e^{\beta_1(t)} + \cdots + b_pe^{\beta_p(t)}}& = \overline{N}^1_1 - t \nu_1(t)\,,\\
...&....\\
\frac{b_pe^{\beta_p(t)}}{b_0 + b_1e^{\beta_1(t)} + \cdots + b_pe^{\beta_p(t)}} &= \overline{N}^1_p - t \nu_p(t)\,,
\end{array}\right.
\end{equation}
and we obtain
\[e^{\beta_i(t)} = \frac{b_0}{b_i}\frac{p_i(t)}{1-\sum_{j=1}^p p_j(t)}\,,\quad i = 1,...,p\,.\]

Let \(\overline{S}_t = \{0,1,...,p\}\setminus S_t\) and \(\overline{S}^*_t =\overline{S}_t \setminus \{0\} \), then
\begin{equation}
\left\{\begin{array}{ll}
e^{\beta_i(t)} = \frac{b_0}{b_i} \frac{\overline{N}^1_i - t {\rm sign}(\beta_i)}{1 - \sum_{j=1}^p \overline{N}^1_j + t\sum_{s \in S_t}{\rm sign}(\beta_s)+ t\sum_{s \in \overline{S}^*_t}\nu_s(t)}\,,\quad & i \in S_t\,,\\
1 = \frac{b_0}{b_i} \frac{\overline{N}^1_i - t \nu_i(t)}{1 - \sum_{j=1}^p \overline{N}^1_j + t\sum_{s \in S_t}{\rm sign}(\beta_s)+ t\sum_{s \in \overline{S}^*_t}\nu_s(t)}\,,\quad & i \in \overline{S}^*_t\,.\\
\end{array}\right.
\end{equation}
After computation, we have explicit formulae for the continuous functions \(\beta_i\) and \(\nu_i\) (\(i \in \{1,...,p\}\)):
\begin{equation}
\left\{\begin{array}{cll}
e^{\beta_i(t)} &= \frac{b^{S}}{b_i} \frac{p_i(t)}{1 - \sum_{s \in S_t}p_s(t)}\,,&\quad i \in S_t\,,\\
\nu_i(t) &= \frac{1}{t} \frac{b^S \overline{N}^1_i - b_i \overline{N}^S}{b^S}-\frac{b_iR^S}{b^S}\,,&\quad i \in \overline{S}^*_t\,,
\end{array}\right.
\end{equation}
with \(b^S = \sum_{s\in \overline{S}_t}b_s\), \(\overline{N}^S = \sum_{s\in \overline{S}_t}\overline{N}^1_s\) and \(R^S = \sum_{s\in S_t} {\rm sign}(\beta_s)\). These functions are monotonous, we need the change-points to draw the path, that is the finite sequence of different models \(\{S_t\}_{t \ge 0} = \{S_{t_0}, S_{t_1},...,S_{t_m}\}\), \(m \in \mathbb{N}^*\). For all \(i \in \{0,...,m-1\}\), \(\{S_t\}_{t \in [t_i,t_{i+1}[}\) is a unique subset.  If \(t \in [t_i,t_{i+1}[\) and we know \(S_t\) we determine \(S_{t_{i-1}}\),\(S_{t_{i+1}}\), \(t_{i}\) and \(t_{i+1}\) by solving
\begin{equation}
\left\{\begin{array}{llll}
\beta_i(u_i) = 0 &\Leftrightarrow  & u_i = \frac{b^S\overline{N}^1_i - b_i \overline{N}^S}{b^S{\rm sign}(\beta_i) + b_i R^S}\,,&\quad i \in S_t\,,\\
\nu_i(v_i^+) = 1 &\Leftrightarrow  & v_i^+ = \frac{b^S\overline{N}^1_i - b_i \overline{N}^S}{b^S + b_i R^S}\,,&\quad i \in \overline{S}^*_t\,,\\
\nu_i(v_i^-) = -1 &\Leftrightarrow  & v_i^- = \frac{b^S\overline{N}^1_i - b_i \overline{N}^S}{-b^S + b_i R^S}\,,&\quad i \in \overline{S}^*_t\,.\\
\end{array}\right.
\end{equation}
We define \(W = \{w_i\} = \{u_i,v_j^+,v_j^-\,,\, i \in S_t\,,\, j \in \overline{S}^*_t\}\) and the two adjacent change-points are given by
\[t_{i+1} = \min_{j}\{w_j\,|\, w_j > t\}\quad {\rm and} \quad t_{i} = \max_{j}\{w_j\,|\, w_j \le t\}\,.\]
Therefore, 
\[S_{t_{i+1}} = S_t \cup V_{i+1}\setminus U_{i+1}\quad {\rm and} \quad S_{t_{i-1}} = S_t \cup V_{i}\setminus U_{i}\,,\]
with \(U_{i} = \{j \in \{1,...,p\} \,|\, u_j = t_{i} \}\,, V_{i} = \{j \in \{1,...,p\} \,|\, v_j^+ = t_{i}\,\,{\rm or}\,\, v_j^-= t_{i} \}\).

The path can be built forward or backward. If we choose the path following approach (forward), 
\(S_{t_0}\) is found using the MLE of the rescaled likelihood (see Section \ref{section3}) and \(t_0 = 0\). In the other configuration (backward), we have \(S_{t_m} = \emptyset \) and for \(t > t_m\), \(b^S=\overline{N}^S = 1\) and \(R^S = 0\), so that \(t_m = \max\limits_{i \in \{1,...,p\}}|\overline{N}^1_i-b_i| \).
\end{proof}
Simulations with this type of design show that each path usually vanishes only one time (and does not reappear) and thus \(m>p\) is a very rare (impossible?) configuration.

The opposite situation to orthogonality is inclusion. For example, if $X_{12}$ is included in $X_{11}$ meaning that for the observed data $\mathtt{I}_{i1}=1$ if $\mathtt{I}_{i2}=1$, we find an analytic description of the estimator given by the formulae
$$\dot{\hat{\beta}}_1(t) = \dot{\overbrace{\log \left( \frac{p_1(t)}{p_2(t)-p_1(t)} \right)}} \,,\, \dot{\hat{\beta}}_2(t) = \dot{\overbrace{\log \left( \frac{p_2(t)-p_1(t)}{1-p_2(t)} \right)}}\,,\quad\forall t \in \{u\,,\hat{\beta}_1(u) \hat{\beta}_2(u) \ne 0\}\,.$$

This solution is likely generalizable (with a design in stairs as presented in Appendix \ref{A4}), however, this case is meaningless in the analysis of spontaneous reports databases and then left aside.\\

We give examples of plots of path estimates compared with a standard (using \(L\) not \(L^*\)) lasso path for different imbalance strengths in appendix~\ref{appC}. The results highlight the high quality of the analytic path estimators, even in absence of class imbalance.

\begin{remark}
Another regularization method is called the elastic net penalization and uses, in addition to the lasso, a second penalized term of ridge (or Tikhonov) kind \cite{Zou}:
$${\cal L} (\beta|y) =  L^*(\beta|y) \times \exp\left(- \lambda \left[\alpha\sum_{i=1}^{p} |\beta_i| + \frac{1-\alpha}{2}\sum_{i=1}^{p} \beta_i^2\right]\right)\,,$$
with $\alpha \in ]0,1]$. In the case of independence in random vector \(X\), we have an explicit formula for $t$ with respect to $\beta$:
$$t = \frac{\overline{N}^1}{\alpha\, {\rm sign}(\hat{\beta}) + (1-\alpha)\hat{\beta}} \frac{1-e^{-\hat{\beta}(0)+\hat{\beta}}}{1+\frac{\overline{N}^1}{1-\overline{N}^1}e^{-\hat{\beta}(0)+\hat{\beta}}}\,,$$
for $\hat{\beta}$ between $0$ and $\hat{\beta}(0)$. The coefficients vanish when $t_0^{en} = \frac{1}{\alpha} t_0^{lasso}\,.$
The proof of this result is a simple adaptation of the proof for the lasso in Theorem \ref{descriptionLasso}.
\end{remark}

\subsection{Negative correlation structure}

If the random design verifies the relations \(\mathbb{E}[X_{1j}X_{1k}e^{X_1\beta}]\mathbb{E}[e^{X_1\beta}] \le \mathbb{E}[X_{1j}e^{X_1\beta}]\mathbb{E}[X_{1k}e^{X_1\beta}]\), \(\forall j \ne k\), \(\forall t \ge 0\), this in-between situation of a \(\beta\)-dependent negative correlation between variables \(X_j\) (\(j=1,...,p\)) is also tractable and particularly interesting in the sparse context of near-zero components for vector \(\overline{N}^1\)\footnote{Spontaneous reports databases are an example of such a sparsity with negative correlation.}. We find two estimators that sourrunded the real path.

\begin{thm}
The path estimator in the \(\beta\)-dependent negative correlation case is surrouned by estimators, whose derivatives are given by
\[\dot{\overbrace{\log \left( \frac{p_j(t)}{1-\sum_{s \in S^+_t}p_s(t)} \right)}} \le \dot{\hat{\beta}}_j(t) \le \dot{\overbrace{\log \left( \frac{p_j(t)}{1-\sum_{s \in S^-_t}p_s(t)} \right)}}\,,\quad j \in S_t\,,\]
with \(S_t^+ = \{j \in S_t\,|\, {\rm sign}(\beta_j)>0\}\) and \(S_t^- = \{j \in S_t\,|\, {\rm sign}(\beta_j)<0\}\). With the rare occurrence of resurgence of a coefficient after vanishing, we neglect this possibility and we easily find the \(p\) vanishing points and thus the family of subsets \(\{S_t\}_{t \ge 0}\). 
\end{thm}

\begin{proof}
We differentiate equations (\ref{LassoIm}) with respect to $t$ considering only the equations verifying the condition \(\beta_j(t) \ne 0\), i.e. \(j \in S_t\). We obtain at time \(t\),
\[\dot{\beta_j}(t) \left(\frac{\mathtt{I}_j^T (\overline{n}^0e^{\mathtt{I}\beta})}{\sum_i \overline{n}^0_ie^{(\mathtt{I}\beta)_i}}-\left(\frac{\mathtt{I}_j^T( \overline{n}^0e^{\mathtt{I}\beta})}{\sum_i \overline{n}^0_ie^{(\mathtt{I}\beta)_i}}\right)^2 \right) +$$
$$ \sum_{k \ne j, k \in S_t}\dot{\beta_k}(t) \left( \frac{\sum_{i}\mathtt{I}_{ij}\mathtt{I}_{ik}\overline{n}^0_ie^{(\mathtt{I}\beta)_i}}{\sum_i \overline{n}^0_ie^{(\mathtt{I}\beta)_i}} -\frac{\mathtt{I}_j^T (\overline{n}^0e^{\mathtt{I}\beta})}{\sum_i \overline{n}^0_i e^{(\mathtt{I}\beta)_i}}\frac{\mathtt{I}_k^T( \overline{n}^0e^{\mathtt{I}\beta})}{\sum_i \overline{n}^0_ie^{(\mathtt{I}\beta)_i}} \right) = -{\rm sign}(\beta_j)\,,\]
or written differently,
\begin{equation}
\label{follpath}
\dot{\beta_j}(t) p_j(t)(1-p_j(t)) + \sum_{k \ne j, k \in S_t}\dot{\beta_k}(t) \left( R_{jk}(t)-p_j(t)p_k(t)\right) = \dot{p}_j(t)\,,
\end{equation}
where $R_{jk}(t) = \frac{\sum_{i}\mathtt{I}_{ij}\mathtt{I}_{ik}\overline{n}^0_ie^{(\mathtt{I}\beta)_i}}{\sum_i \overline{n}^0_ie^{(\mathtt{I}\beta)_i}}$ is a t-dependent proportion of rows with a one on the columns $j$ and $k$. With only negative correlations or independence between components of \(X\), we define the matrix \(F(t) \in \mathcal{M}_{r \times r}([0,1])\) with \(r = \#S_t\) as long as \(R_{jk}(t) \le p_j(t)p_k(t)\),
\[R_{jk}(t) - p_j(t)p_k(t)= (F_{jk}(t)-1)p_j(t)p_k(t)\,,\]
if observations \(\mathtt{I}\) give such a matrix \(F\). We obtain \((\mathbb{I}_r - (D-F)P)\dot{\beta} = P^{-1}\dot{P}(1)\) with $P$ a diagonal matrix filled with the elements $\{p_i(t)\,,\, i \in S_t\}$. Matrix \(D\) is the correlation-track matrix containing ones at positions \((j,k)\) if \(F_{jk}(0)<1\) and we have\footnote{the non-singularity of the matrix \(C(t)\) in (\ref{follpath}), \(C(t)\dot{\beta} = \dot{\overbrace{\log p(t)}}\), will be proven with Proposition \ref{inversibility}.}
\[\dot{\beta}(t) = \sum_{i=0}^{+\infty} ((D-F)P)^i P^{-1}\dot{P}(1) = P^{-1}\dot{P}(1) +  \sum_{i=1}^{+\infty} ((D-F)P)^{i}P^{-1}\dot{P}(1)\,,\]
so that, using the positivity of all the elements in matrix \(D-F\):
\[ \dot{\overbrace{\log(P(1))}} - \dot{\overbrace{\log(1-DP^+(1))}} \le \dot{\beta}(t) \le  \dot{\overbrace{\log(P(1))}} - \dot{\overbrace{\log(1-DP^-(1))}}\,,\]
with \(P^+\) the diagonal matrix filled with vector \(p^+(t)= (\max(sign(\beta_i),0)) p_i(t))_{i \in S_t}\) and \(P^-\) with vector \(p^-(t)=(\max(-sign(\beta_i),0)) p_i(t))_{i \in S_t}\). Finally,
\[\dot{\overbrace{\log \left( \frac{p_j(t)}{1-(Dp^+(t))_j} \right)}} \le \dot{\hat{\beta}}_j(t) \le \dot{\overbrace{\log \left( \frac{p_j(t)}{1-(Dp^-(t))_j} \right)}}\,,\quad j \in S_t\,.\]
\end{proof}

In presence of sparsity (small components in \(\overline{N}^1\)), \(0 < p_j(t) \ll 1-(Dp^-(t))_j(t)\) and \(0 < p_j(t) \ll 1-(Dp^+(t))_j(t)\), which makes previous upper and lower bounds good path estimators. The \(p\) (or more) change-points are determined step by step as in previous subsection and the estimated path \(\beta(t)\) is stucked between a lower path and an upper paths.

\section{Efficient algorithms for Lasso regularization}

In this last section, we propose two new algorithms drawing piecewise logarithmic approximate paths derived from a small amount of matrix inversions (\(p\) or more). The logarithmic function naturally arised in the expression of all previously found path estimators, consequently, we build approximations involving this function. The main benefit of our algorithms is the direct computation of the sequence \(\{t_i\}\) as done by the LARS \cite{Efron} for linear regression. Our first algorithm follows the path (\(t\) increases) and is a simplified procedure adapted to data with a low correlation structure. The second algorithm is a backward procedure (\(t\) decreases toward zero) and can challenge the classic coordinate descent approach \cite{Fri1}. The efficiency of the algorithms are eventually illustrated on pharmacovigilance data.

\subsection{Cauchy problem}

The derivative of the first order equations for the Lasso with respect to \(t\) leads to a Cauchy problem.

\begin{proposition}
\label{inversibility}
The Lasso regularization path is described by the following system of differential equations
\[\dot{\beta}(t) = C(t)^{-1}\,\dot{\overbrace{\log p(t)}}\,, \quad t >0 \,, \]

with \(C(t) \in \mathcal{M}_{r_t \times r_t}(\mathbb{R})\) (\(r_i = \#S_{t}\)), \(\beta \in \mathbb{R}^{r_t}\), \(\log p(t) \in \mathbb{R}^{r_t} \) and
\[C_{jk}(t) = \left( \frac{\sum_{u}\mathtt{I}_{uj}\mathtt{I}_{uk}\overline{n}^0_ue^{(\mathtt{I}\beta)_u(t)}}{\mathtt{I}_j^T( \overline{n}^0e^{\mathtt{I}\beta(t)})} -p_k(t) \right)\,, \quad j,k \in S_{t} \subset \{1,...,p\}\,.\]

\end{proposition}

\begin{proof}
Equations (\ref{follpath}) are divided by vector \(p(t)\) and we obtain the desired equations. It remains to be proven the non-singularity of matrix \(C(t)\) for all \(t>0\). \\
With diagonal matrix \(P  \in \mathcal{M}_{r_t \times r_t}(\mathbb{R}) \) filled by elements \((p_j)_{j\in S_{t}}\) we build a matrix \(\tilde{C} = PC\) whose elements are:
\[(PC)_{jk} = \tilde{C}_{jk}= \frac{\sum_{u}\mathtt{I}_{uj}\mathtt{I}_{uk}\overline{n}^0_ue^{(\mathtt{I}\beta)_u}}{ \sum_u\overline{n}_u^0e^{(\mathtt{I}\beta)_u}} -p_j(t)p_k(t) \,, \quad \quad j,k \in S_{t}\,.\]
Suppose that this matrix \(\tilde{C}(t)\) is singular, then there exists a non-identically null vector \(\gamma \in \mathbb{R}^{r_t}\) such that \(\tilde{C}(t)\gamma = 0\) or written component-by-component 
\[(\tilde{C}\gamma)_j = \frac{\sum_{u}\mathtt{I}_{uj}(\sum_k \mathtt{I}_{uk} \gamma_k)\overline{n}^0_ue^{(\mathtt{I}\beta)_u}}{ \sum_u\overline{n}_u^0e^{(\mathtt{I}\beta)_u}} -p_j(t)\frac{\sum_{u}(\sum_k \mathtt{I}_{uk} \gamma_k)\overline{n}^0_ue^{(\mathtt{I}\beta)_u}}{ \sum_u\overline{n}_u^0e^{(\mathtt{I}\beta)_u}} =0 \,,\quad j \in S_t\]
We compute the linear combination \(\sum_j \gamma_j (\tilde{C}\gamma)_j = 0\) to obtain after computations
\[(\sum_{u}J_u^2\overline{n}^0_ue^{(\mathtt{I}\beta)_u}) (\sum_u\overline{n}_u^0e^{(\mathtt{I}\beta)_u}) -(\sum_{u}J_u\overline{n}^0_ue^{(\mathtt{I}\beta)_u})(\sum_{u}J_u\overline{n}^0_ue^{(\mathtt{I}\beta)_u}) =0 \,,\]
with \(J_u =  \sum_l \mathtt{I}_{ul} \gamma_l\). This relation is expanded and simplified into
\[\sum_{u \ne v}(J_u-J_v)^2 \overline{n}^0_u \overline{n}^0_v e^{(\mathtt{I}\beta)_u+(\mathtt{I}\beta)_v} =0 \,.\]
This is a sum of positive terms equals to zero, meaning that each term wanishes and we get \(J_u = const\) for all \(u = 1,...,n\). Thus \(\mathtt{I}\gamma = const\) which is impossible because matrix \(I\) is a full rank matrix. 
\end{proof}

\subsection{The piecewise logarithmic approximate path : a first simple algorithm}
\label{subsection44}

Path following algorithms \cite{Rosset} are competing methods with more used coordinate descent algorithms \cite{Fri1} \cite{Fri2}. We here present a simple algorithm for an increasing regularization parameter \(t\). Within this procedure, we are able to estimate at each step the value \(t\) of the next wanishing component in vector \(\beta(t)\) and thus speeding up the classical Newton-Raphson step \cite{Rosset}. We consider that correlation between predictors is "low", so that an emergence of a coefficient along the path after wanishing is not taken into account (but this case is included in the second algorithm). \\ 

\begin{proposition}
\label{algo1}
The path following algorithm for limit imbalanced logisitic regression by binary predictors (with low correlation) is the following:\\
\(i=0\), \(t_0=0\), \(\beta(t_0)=\beta(0)\) given. \(S_0 = \{j\,|\, \beta_j(0) \ne 0\,,\, j = 1,...,p\}\).\\
WHILE \(r_i = \#S_{t_i} \ne 0\) DO
\begin{equation}
\label{tseq}
t_{i+1} = t_i + \min \Delta T_i\,,\quad \Delta T_i = \{\Delta t_j\,|\, \Delta t_j = \frac{1-e^{-\beta_j(t_i)}}{(C_i^{-1}(\frac{{\rm sign}(\beta)}{p(t_i)}))_j} > 0\,,\, j \in S_{t_i}\}\,,
\end{equation}
with \(C_i  \in \mathcal{M}_{r_i \times r_i}(\mathbb{R}) \) such that
\[(C_i)_{jk}(t_i) = \left( \frac{\sum_{u}\mathtt{I}_{uj}\mathtt{I}_{uk}\overline{n}^0_ue^{(\mathtt{I}\beta)_u(t_i)}}{\mathtt{I}_j^T( \overline{n}^0e^{\mathtt{I}\beta(t_i)})} -p_k(t_i) \right)\,, \quad j,k \in S_{t_i} \subset \{1,...,p\}\,.\]
The path, on the segment \([t_i, t_{i+1}]\), is given by 
\[\beta(t)-\beta(t_i) = \log\left( 1 -C_i^{-1} \left(\frac{{\rm sign}(\beta)}{p(t_i)}\right)(t-t_i)\right)\,,\quad t \in [t_i, t_{i+1}]\,,\]
\[S_{t_{i+1}} = S_{t_i} \setminus U_i\,,\quad U_i = \Big\lbrace  j \in S_{t_i} \,|\, \Delta t_j = \min \Delta T_i \Big\rbrace\,.\]
\(i\) becomes \(i+1\).\\
END DO.
\end{proposition}

\begin{proof}
Equations (\ref{follpath}) take the form \(C(t)\dot{\beta} = \dot{\overbrace{\log \left(p(t)\right)}}\) with \(C(t)\) called correction matrix. 
\[C_{jk}(t) = \left( \frac{\sum_{u}\mathtt{I}_{uj}\mathtt{I}_{uk}\overline{n}^0_ue^{(\mathtt{I}\beta)_u}}{\mathtt{I}_j^T( \overline{n}^0e^{\mathtt{I}\beta})} -p_k(t) \right)\,, \quad \quad j,k \in S_{t}\,.\]
Between two annulations of regression coefficients along the path (\(t_i\) and \(t_{i+1}\)), we consider this matrix to be constant (\(C(t_i)=C_i\)). In this case, 
\[\beta(t_{i+1})-\beta(t_i)=C_i^{-1}\left[\log(p(t_{i+1}))-\log(p(t_{i}))\right]\,.\]
We have \(t_0=0\), but the sequence of values \(\{t_i\}\) is unknown. However, we iteratively approximate them as follows. With
\[\beta(t_{i+1})-\beta(t_i)=C_i^{-1}\log\left( 1 - \left(\frac{{\rm sign}(\beta)}{p(t_i)}\right)(t_{i+1}-t_i)\right)\]
\begin{equation}
\approx \log\left( 1 -C_i^{-1} \left(\frac{{\rm sign}(\beta)}{p(t_i)}\right)(t_{i+1}-t_i)\right)\,,\label{approxx}
\end{equation} 
because \(|C_i^{-1} \left(\frac{{\rm sign}(\beta)}{p(t_i)}\right)(t_{i+1}-t_i)|\) is small for relative small step \(t_{i+1}-t_i\). We obtain the piecewise logarithmic path:
\[\beta(t)-\beta(t_i) = \log\left( 1 -C_i^{-1} \left(\frac{{\rm sign}(\beta)}{p(t_i)}\right)(t-t_i)\right)\,,\quad t \in [t_i, t_{i+1}]\,,\]
with
$$t_{i+1} = t_i + \min \Delta T_i\,,\quad \Delta T_i = \{ \Delta t_j\,|\, \Delta t_j = \frac{1-e^{\beta_j(t_i)}}{(C_i^{-1}(\frac{{\rm sign}(\beta)}{p(t_i)}))_j} > 0\,,\, j \in S_{t_i}\}\,,\quad i = 0,1,...$$
\(\Delta T_i\) is the set of values for \(t_{i+1}-t_i\) solving (\ref{approxx}) with \(\beta_j(t_{i+1})=0\) (for each \(j \in S_{t_i}\)). The set \(U_i\) gives at each step the indexes of regression coefficients to remove from \(S_{t_i}\). 
\end{proof}

Other approximations could be performed, for example using a second order term in the previous approximation (\ref{approxx}). Simulation tests show that our choice seems to give better results. We notice that the size of the matrix $C_i$ decreases during this procedure, speeding up the computation at each new step \(t_i\).

\begin{remark}
This algorithm has two main computational advantages. Firstly, the sequence \(\{t_i\}_{i=1,...,p}\) is directely determined, whereas other algorithms use a regular discretization on a logarithmic scale (coordinate descent) or Newton-Raphson steps (path following). Secondly, the sum \(\sum_i \overline{n}^0_ie^{(\mathtt{I}\beta)_i}\) does not appear in the $C_i$ matrices, which can highly reduce the computational cost especially if the matrix $\mathtt{I}$ is sparse (\(0.03\%\) of ones in the French spontaneous reports data base): this algorithm handles sparsity!
\end{remark}

To explore the efficiency of the algorithm, we simulate data sets with different correlation structures. Model selection is often provided with the BIC \cite{Schwarz}, which requires to know the different models arising along the path. Hence, we decide to evaluate the algorithm accuracy using a simple indicator: a comparison of the sequence of coefficients in the order of wanishing along the path. The indicator is \(p'/p\) if a simulation with our algorithm gives \(p'\) coefficients at the same index as in the sequence obtained by a classical lasso algorithm (coordinate descent in R package \emph{glmnet}). The correlation coefficient (from \(r= 0\) to \(r = 0.9\)) means that we chose initial \( R_{jk}=(1-r)b_j  b_k+r \min(b_j,b_k)\).\\
We simulate \(10^3\) paths for each number \(nb\) and \(r\), \(nb\) being the number of predictors in correlation. For each path, \(\beta_0=-5\) and the \(10\) regression coefficients \((\beta_1,...,\beta_{10})\) are always the same and chosen on a regular scale between \(-0.5\) and \(0.5\).

\begin{table}[!ht]
\begin{tabular}{c|cccccccccc}
nb/r&0&0.1&0.2&0.3&0.4&0.5&0.6&0.7&0.8&0.9\\
\hline
3 (i)&0.977& 0.909& 0.862& 0.809& 0.740& 0.706 &0.678& 0.642& 0.620& 0.564\\
3 (a)&0.777& 0.710& 0.736& 0.730& 0.735& 0.720& 0.727& 0.726& 0.725 & 0.755\\
\hline
5 (i)&0.970& 0.859& 0.771& 0.671& 0.613& 0.552& 0.513& 0.488& 0.410&  0.361\\
5 (a)&0.785& 0.752& 0.725& 0.715& 0.700& 0.760& 0.708& 0.693& 0.699& 0.705\\
\hline
8 (i)&0.969& 0.868& 0.750& 0.609& 0.550& 0.486& 0.433& 0.363& 0.341& 0.309\\
8 (a)&0.773& 0.742& 0.717&  0.756& 0.722& 0.698& 0.714& 0.675& 0.684& 0.622\\
\hline
\end{tabular}
\caption{Path model sequence analysis. The simple algorithm (a) has a stronger robustness to the presence of correlation than the analytic path obtained with an asumption of independence (i).}
\label{table4}
\end{table}

With \(nb = 3, 5\) or \(8\) correlated predictors over the 10 used, the exact solution with assumption of independence (i) (see Theorem \ref{descriptionLasso}) deteriorates with the increase in correlation (r), which is (almost) not the case if we use our algorithm (a). Notice that, with a result around \(0.8\), the approximate path is often very close to the exact one, this is due to the inversion in the sequence of two close \(t_i\) terms (see (\ref{tseq})).

\subsection{A new algorithm}

The second algorithm presented in this section computes forward selection. It is more suitable for problems with a large number of predictors (when we are looking for a sparse model) or/and in presence of a strong correlation structure. \\

The standard approach for computing regularization path by decreasing \(t\) with logistic regression consists in using a first order quadratic approximation of the first derivative of the likelihood between two consecutive closed solutions (that is in practice, two parameters \(t_i\) and \(t_{i+1}\) such that \(t_{i+1}-t_i<0\) is small). Using small steps for the parameter sequence \(\{t_i\}\) to ensure a good approximation, the path is drawn by the cyclical coordinate method (see \cite{Fri1} and the R package {\it glmnet}). Our new algorithm is a kind of equivalent of the LARS algorithm for the logistic regression : we compute large step in \(t\). Furthermore, in comparison with the cyclic coordinate descent algorithm, there is no loop at a fixed parameter \(t\). After presenting the algorithm, we challenge the {\it glmnet} package with our approach. \\

\begin{proposition}
The backward algorithm for limit imbalanced logisitic regression by binary predictors is the following:\\
\( i=0\), \(t_0=\max_i\{|\overline{N^1_i}-\overline{N^0_i}|\} = |\overline{N^1_k}-\overline{N^0_k}| \), \(\beta(t_0)=(0,...,0)^T \in \mathbb{R}^p\) and \(\epsilon>0\) given. \(S_{t_0} = \{\beta_k\}\).\\
WHILE (\(t_i>\epsilon\) or \(\#S_{t_i}<p\)) DO
\[t_{i+1} = t_i + \max \{\Delta T_i, \overline{\Delta T_i}\}\,,\]
with
 \[ \Delta T_i = \lbrace \Delta t_u \,|\, \Delta t_u= \frac{1-e^{-\beta_u(t_i)}}{\Phi_{iu}}<0\,,\, u \in S_{t_i}\rbrace\]
and
\[\overline{\Delta T_i} = \lbrace \overline{\Delta t_j^+},\,\overline{\Delta t_j^-} \,|\,\overline{\Delta t_j^+} = t_i \frac{1 - \nu_j(t_i)}{-1 + \Psi_{ij}p_j(t_i)}<0\,,\,\overline{\Delta t_j^-} = t_i \frac{1 + \nu_j(t_i)}{-1 - \Psi_{ij}p_j(t_i)}<0\,,\, j \in \overline{S}^*_{t_i}\rbrace\,.\]
Definitions for matrices \(\Phi\) and \(\Psi\) are given in the proof. Notice that \(\Phi = \Phi(t_i,\beta(t_i))\) (as for \(\Psi\)). 
The path, on the segment \([t_{i+1}, t_i]\), is given by
\[\beta_j(t)-\beta_j(t_i) = \log\left( 1 -\Phi_{ij}(t-t_i)\right)\,,\quad t \in [t_{i+1},t_i]\,,\quad j \in S_{t_i}\,,\]
and for the subgradients
\[\nu_j(t) = \frac{t_i}{t}\nu_j(t_i) + (1- \frac{t_i}{t})\Psi_{ij}p_j(t_i)\,,\quad t \in [t_{i+1},t_i]\,,\quad j \in \overline{S}^*_{t_i}\,.\]
The new set \(S_{t_{i+1}}\) is given by
\[S_{t_{i+1}} = (S_{t_i} \setminus U_i) \cup U'_i \,,\]
with
\[U_i = \Big\lbrace  j \in S_{t_i} \,|\, \Delta t_j = \max \{\Delta T_i, \overline{\Delta T_i}\} \Big\rbrace\,,\quad U'_i = \Big\lbrace  u \in \overline{S}^*_{t_i} \,|\,  \overline{\Delta t_u^+} \,{\rm or}\, \overline{\Delta t_u^-} = \max \{\Delta T_i, \overline{\Delta T_i}\} \Big\rbrace\,.\]
\(i\) becomes \(i+1\).\\
END DO
\end{proposition}

\begin{proof}
We differentiate equations (\ref{LassoIm}) for all \(j\) in \(\{1,...,,p\}\) (see also (\ref{follpath})):

\[ \sum_{k \in S_t} \left( \frac{\sum_{i}\mathtt{I}_{ij}\mathtt{I}_{ik}\overline{n}^0_ie^{(\mathtt{I}\beta)_i}}{\mathtt{I}_j^T (\overline{n}^0e^{\mathtt{I}\beta})} -p_k(t) \right)\dot{\beta_k}(t) = \dot{\overbrace{\log(p_j(t))}}\,,\]
or in matrix form with \(C(t) \in \mathcal{M}_{r \times r}(\mathbb{R})\), \(D(t) \in \mathcal{M}_{(p-r) \times r} (\mathbb{R})\), \(r = \#S_t\) and vectors \(p^{\ne}(t) = (p_j(t))_{j\in S_t}^T\), \(p^{=}(t) = (p_j(t))_{j \in \overline{S}_t}^T\) and \(\beta^{\ne}(t) = (\beta_j(t))_{j\in S_t}^T\) we get
\[C(t)\dot{\overbrace{\beta^{\ne}(t)}} = \dot{\overbrace{\log(p^{\ne}(t))}}\,,\quad D(t)\dot{\overbrace{\beta^{\ne}(t)}} = \dot{\overbrace{\log(p^{=}(t))}}\,.\]
\(C(t)\) is a square non-singular matrix for all \(t\) in \([0,t_0[\) (see Remark \ref{inversibility}). Between two consecutive values \(t_i\) and \(t_{i+1}\) (\(t_{i+1}<t_i\)) of the \(t\) sequence, we consider that \(C(t) \approx C(t_i)\) and \(D(t) \approx D(t_i)\), thus
\[\dot{\overbrace{\beta^{\ne}(t)}} \approx C^{-1}(t_i)\dot{\overbrace{\log(p^{\ne}(t))}}\,,\quad \dot{\overbrace{\log(p^{=}(t))}} \approx D(t_i)C^{-1}(t_i)\dot{\overbrace{\log(p^{\ne}(t))}} = E(t_i)\dot{\overbrace{\log(p^{\ne}(t))}}\,,\]
with \(E(t_i) \in \mathcal{M}_{(p-r_i) \times r_i} (\mathbb{R})\) and \(r_i = \#S_{t_i}\). The system of equations involving matrix \(C^{-1}\) is solved as in the proof of Proposition \ref{algo1} and we get 
\[\beta_j(t)-\beta_j(t_i) = \log\left( 1 -\Phi_{ij}(t-t_i)\right)\,,\quad t \in [t_{i+1},t_i]\,,\quad j \in S_{t_i}\,,\]
with \(\Phi_{ij} = \left(C_i^{-1}(\frac{{\rm sign}(\beta^{\ne})}{p^{\ne}(t_i)})\right)_j\). The second set of equations gives
\[\log(p^{=}(t))-\log(p^{=}(t_i)) = E_i(\log(p^{\ne}(t)-\log(p^{\ne}(t_i))\,,\]
and using the usual approximation
\[\log(p_j(t))-\log(p_j(t_i)) \approx \log\left(1- \Psi_{ij}(t-t_i)\right)\,,\quad j \in \overline{S}^*_{t_i}\,,\]

with \(\Psi_{ij} = \left(E_i(\frac{{\rm sign}(\beta^{\ne})}{p^{\ne}(t_i)})\right)_j\) and we find
\[\nu_j(t) = \frac{t_i}{t}\nu_j(t_i) + (1- \frac{t_i}{t})\Psi_{ij}p_j(t_i)\,,\quad t \in [t_{i+1},t_i]\,,\quad j \in \overline{S}^*_{t_i}\,.\]
We solve \(2 r_i +(p-r_i) = p + r_i\) equations (\(\nu_j(t_{i+1})=\pm 1\), \(j \in \overline{S}^*_{t_i}\) and \(\beta_j(t_{t+1})=0\), \(j \in S_{t_i}\)) to find the possible values for \(t_{i+1}-t_i\). The maximum of obtained negative values within the \(p + r_i\) results is used to build the \(t\) sequence.
\end{proof}
To visualize what is happening during the algorithm, we define linear functions \( B^{=}_j : t \mapsto t\nu_j(t)\) and \(B^{\ne}_j : t \mapsto e^{\beta_j(t)}-1+{\rm sign}(\beta_j)t\)  leading to the $p$ functions \(B_j\) (\(j=1,...,p\)) such that
\[B_j(t) = \left\{\begin{array}{cll}
B_j^=(t)\,,&{\rm if}&|B_j^=(t)|\le t\,,\\
B_j^{\ne}(t)\,,&{\rm if}&|B_j^{\ne}(t)| > t\,.\\
\end{array}\right.\]
Functions \(B_j\) are all piecewise linear and can be drawn in the plane shown in Figure \ref{nubeta}.
\begin{figure}[!ht]
\center
\includegraphics[totalheight=0.24\textheight]{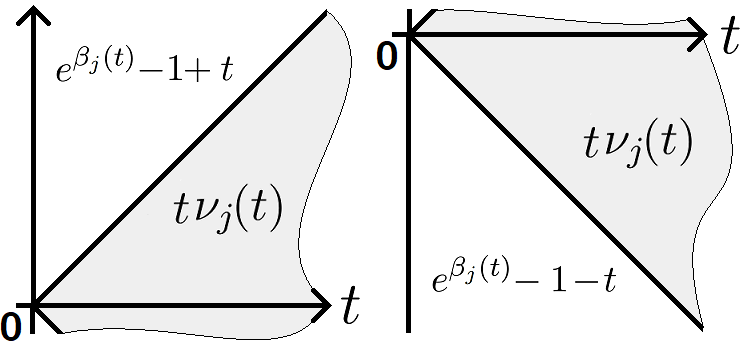} 
\caption{The \(B_j\) functions are piecewise linear in this plane. When \(B_j(t) = \pm t\), \(t\) is one of the value of the sequence \((t_i)_{i \ge 1}\).}
\label{nubeta}
\end{figure}

\subsection{Path reconstruction with the French spontaneous reports database}

We illustrate the efficiency of the limit path construction by piecewise logarithmic functions on the French spontaneous reports database. We look at two examples, a first one with no evidence of correlation and a second with strong correlations. The database contains about 330000 reports in 2016 and the imbalance is high or very high for all the adverse effects \cite{Beziz}. In the following graphs, the dotted lines represent results obtained by our algorithm, the solid ones result from the classical \emph{glmnet} package.

\begin{figure}[!ht]
\center
\includegraphics[totalheight=0.3\textheight]{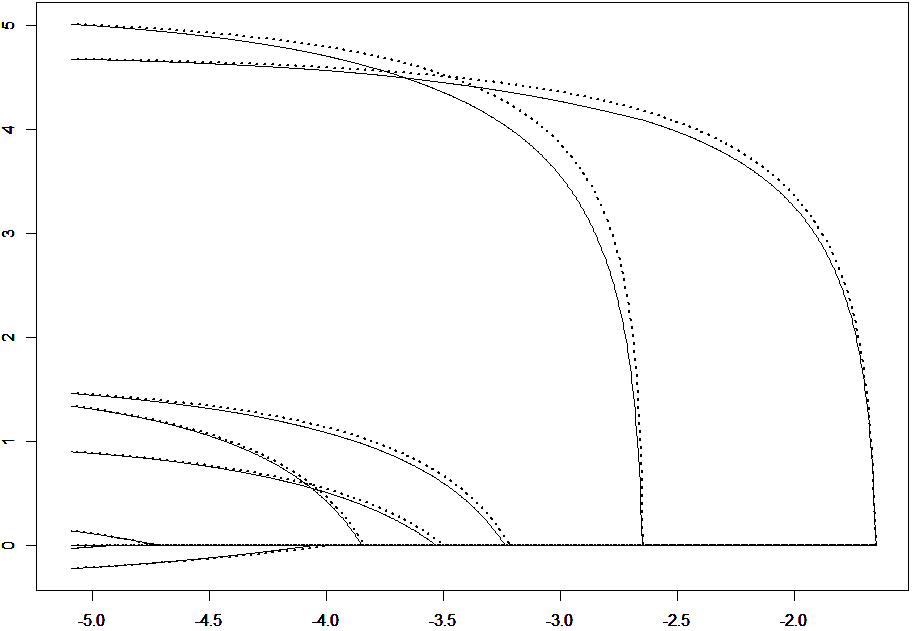} 
\caption{Example 1. Path of regression coefficients with respect to log(t). The correlation between predictors is weak and the path given by the exact solution with independent predictors is a good approximation of the result.}
\label{F1}
\end{figure}

\begin{figure}[!ht]
\center
\includegraphics[totalheight=0.24\textheight]{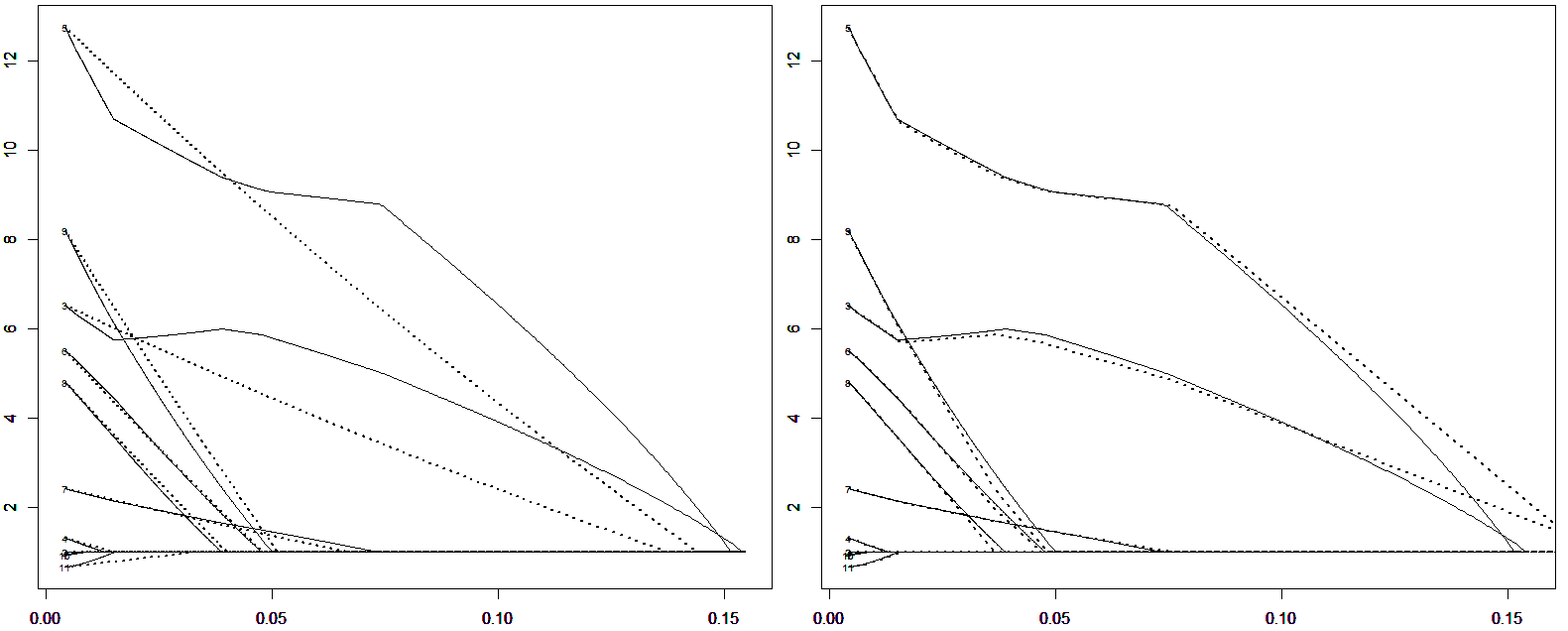} 
\caption{Example 2. We plotted the exponential of the beta coefficients with respect to t. On the left side by the assumption of no correlation, on the right side with the correction due to the approximate path algorithm. In presence of stong correlations, the algorithm of the approximate path (described in previous subsection) gives better adjusted graphs.}
\label{F2}
\end{figure}

The Figure \ref{F2} shows common features encountered with other examples. The path of the exponential of the coefficients shapes a set of piecewise linear functions and the algorithm remains efficient even if the number of predictors is high (150 for examples). It seems that there is no case of a path with a curve reappearing after a first canceling (due to a strong correlation between predictors with opposite signs of initial coefficients). Thus, the sets \(A_i\) in the algorithm do not have to be determined.\\

We notice that the accuracy of this path following algorithm can easily be increased by adding intermediate steps (in variable t). The main computing limitation being the matrix inversion, one could study the inner product (Gram) matrix for class 0 and reorder rows and columns to reveal patterns and form a block diagonal matrix. These blocks could result from a statistical study of the Gram matrix\footnote{To that end, see the literature of the block clustering problem \cite{Govaert}.} (finding the pairwise independent predictors) as well as from pharmacological assumptions (medical treatments also shape patterns). Thus, computational costs become a marginal problem and one can concentrate on the bias correction by adding priors related to temporal bias, under-reporting or the introduction of similarity modifying the \(R\) matrix\footnote{With a similarity matrix \(S \in \mathcal{M}_{p \times p}([0,1])\), the similarity is defined as follows: coefficient \(R_{jk}\) becomes \((1-S_{jk})R_{jk} + S_{jk}\min(p_j,p_k)\).}.

\section{Conclusion and perspectives}
The central novelty of this work is the introduction of a rescaled likelihood for the limit imbalanced logistic regression problem. The expression of this likelihood could have some connexions with the well-known likelihoods of the self-controlled case series method \cite{Simpson} and of the proportional hazards model \cite{Cox} used in epidemiology.

Most results exposed for binary data can be extented to other data types. However, simulations have been done only with binary data, having in mind the underlying applied problem of pharmacovigilance. The new estimate is always very close to the initial MLE because data are located on the vertexes of the hypercube and then one another "close". A convergence study of all possible existing algorithms for the primal and dual problems could be performed with different class imbalances and an evaluation of the first order term.

The variance reduction is a central issue that has to be treated in a Bayesian framework. Whereas the prior to add in the standard logistic regression is unclear, the rescaled likelihood takes a well-adapted form for exponential priors. We considered model selection using the BIC and the lasso to answer this question. Due to binary data, the lasso regularization problem became easier to understand in our limit imbalanced case: we found many precise estimators. Piecewise logarithmic approximate paths are built by an effective path following procedure which determines step by step the vanishing time of each path, do not use any loops as in coordinate descent algorithms and computes expressions only involving non-zero data. Moreover, this algorithm can take into account the correlation structure between predictors to further shrink computational costs. The values for \(\overline{N}^1\), \(\overline{n}^0\) and for matrix \(R\) could be shifted in order to incorporate absolute bias, temporal bias, under-reporting and similarity or correlation corrections.

\section{A pharmacovigilance project?}

Within this paper, we have had in mind the pharmacovigilance context as this work was carried out in parallel of a one-year engineering job at the French National Institute of Health and Medical Research\footnote{B2PHI laboratory UMR 1181, INSERM, UVSQ, Institut Pasteur, Villejuif 94807, France}. We hope this article could contribute a little to the developement of mathematical tools for pharmacovigilance purposes. The science of drug safety at a postmarketing level is nearly non-existent in France as in many other countries: the reporting process of spontaneous reports is inadequate and resulting databases are badly processed with unadapted tools. Public health scandals related to medication are steadily increasing and the spotlights are turned towards big pharmaceutical companies while patient associations should firstly require public authorities to establish a modern drug safety structure. To that end, the statistical community has a major role to play by proposing trustworthy decision-support tools, opposing science to political and financial influences. Creating a useful tool was the guideline of this present work and the author hopes that other mathematicians will embrace the direction initiated by this article. \\

We would like to conclude by giving our opinion about the work that remains to be done to obtain an operational tool (in five points), hoping that it will inspire epidemiologists.

1) Building priors related to bias (temporal bias, under-reporting...) with the help of pharmacologists. 2) Developing the proposed regularization algorithms evaluating their complexity and accuracy levels. 3) Introducing simple indicators to control the quality of the limit approximation. 4) Working on path visualization and new indicators (that are not thresholds). 5) Evaluating the obtained tool in the hands of pharmacologists (the use of reference sets is, to our mind, inadequate).

\section*{Acknowlegment}
I would like to deeply thank Laetitia Comminges from the Paris-Dauphine University for relevant comments that greatly improved the manuscript. I also thank my colleague Mohammed Sedki from the INSERM laboratory of Villejuif for his constant encouragement to complete this work.

\appendix
\section{Exact solutions}\label{appA}

We give a collection of examples consisting in simple solutions of the equation (\ref{solution}). 

\subsection{No intercept}
If there is no intercept and no interaction between the regressors, the matrix $\mathrm{I}$ equals the identity matrix $\mathbb{I}_p$ and
\[\beta_1 = \log\left(\frac{n_1^1}{n_1^0} \right)\,,\quad ...\quad ,\, \beta_p = \log\left(\frac{n_p^1}{n_p^0} \right)\,. \]
If one row contains other ones, the inverse matrix is the same matrix with the added ones transformed into its opposite. 
      
\subsection{Intercept}
\label{A2}
If the square matrix $\mathrm{I}_{p+1}$ is the following 
$$\mathrm{I}_{p+1} =  \begin{pmatrix}
   1& 0 & \cdots & \cdots & 0 \\
   \vdots& 1 & 0 & \cdots & 0 \\
   \vdots& 0 & \ddots& \ddots &  \vdots \\
   \vdots& \vdots & \ddots & \ddots&  0 \\
   1& 0 & \cdots& 0 & 1 \\   
   \end{pmatrix}\,,\quad \hbox{then} \quad \mathrm{I}_{p+1}^{-1} =  \begin{pmatrix}
   1& 0 & \cdots & \cdots & 0 \\
   -1& 1 & 0 & \cdots & 0 \\
   \vdots& 0 & \ddots& \ddots &  \vdots \\
   \vdots& \vdots & \ddots & \ddots&  0 \\
   -1& 0 & \cdots& 0 & 1 \\   
   \end{pmatrix}$$
for the inverse matrix, so that the $\beta$ coefficients take the form
$$\beta_0 = \log\left(\frac{n_0^1}{n_0^0} \right)\,,\quad \beta_1 = \log\left(\frac{n_1^1}{n_1^0}\frac{n_0^0}{n_0^1}  \right)\,,\quad ...\quad ,\,\beta_p = \log\left(\frac{n_p^1}{n_p^0}\frac{n_0^0}{n_0^1}  \right)\,. $$

\subsection{Intercept with one correlation}

The first row of the following matrix
$$\mathrm{I}_{p+1} = 
  \begin{pmatrix}
   1& ? & \cdots & \cdots & ? \\
   \vdots& 1 & 0 & \cdots & 0 \\
   \vdots& 0 & \ddots& \ddots &  \vdots \\
   \vdots& \vdots & \ddots & \ddots&  0 \\
   1& 0 & \cdots& 0 & 1 \\   
   \end{pmatrix}$$
definies the set $K = \{j \in \{0,...,p\}\,,\, \mathrm{I}_{1j}=1 \}$. The case $\#J = 2$ is left out because it does not coincide with a non-singular matrix. The case  $\#J = 1$ corresponds to the previous example. The easiest way to solve this example is to look at initial equations (\ref{normal}). We write down the $p+1$ equations, where only the first one has a different form:

$$\sum_{i=0}^{p} (n_i^1-n_i^0) = \sum_{i=1}^{p}(n_i^1+n_i^0)\tanh(\frac{\beta_0+\beta_i}{2})+ (n_0^1+n_0^0)\tanh(\frac{1}{2}\sum_{j \in K} \beta_j) $$
and for $k \in \{1,...,p\}$,
\begin{equation}
\label{p}
(n_k^1-n_k^0) +  1_{k \in K}(n_0^1-n_0^0) = (n_k^1+n_k^0)\tanh(\frac{\beta_0+\beta_k}{2}) + 1_{k \in K}(n_0^1+n_0^0)\tanh(\frac{1}{2}\sum_{j \in K} \beta_j)\,.
\end{equation}

Subtracting all the $p$ equations to the first one, we obtain
$$(1-\#K) (n_0^1-n_0^0) = (1-\#K)(n_0^1+n_0^0)\tanh(\frac{1}{2}\sum_{j \in K} \beta_j),\,$$
which can be used to simplify equations (\ref{p}) into
$$(n_k^1-n_k^0) = (n_k^1+n_k^0)\tanh(\frac{\beta_0+\beta_k}{2})\,.$$
Finally, we have
$$\exp(\sum_{j \in K} \beta_j) = \frac{n_0^1}{n_0^0}\quad \hbox{and}\quad\exp(\beta_0+\beta_k) = \frac{n_k^1}{n_k^0}\,,\quad k \in \{1,...,p\}\,,$$
so that we deduce the following closed form for the $\beta$ coefficients

$$\exp(\beta_0) = \left( \frac{n_0^0}{n_0^1}\prod_{j \in J,\, j \ne 0}\left(\frac{n_j^1}{n_j^0} \right)\right)^{\frac{1}{\#J-2}} \,,$$

$$\exp(\beta_i) = \frac{n_i^1}{n_i^0}\left(\frac{n_0^1}{n_0^0}\prod_{j \in J,\, j \ne 0}\frac{n_j^0}{n_j^1} \right)^{\frac{1}{\#J-2}}\,,\quad i \in \{1,...,p\}\,. $$

Notice that the regression coefficients behave in a very unpredictable way. It is sufficient to see that on an example with $p=2$ and $\#J =3$. The matrix $\mathrm{I}$ is 
$$
  \begin{pmatrix}
   1& 1 & 1 \\
   1& 1 & 0 \\
   1& 0 & 1 \\ 
   \end{pmatrix}\,,$$
and we have

$$\exp(\beta_0)=\frac{n_0^0 }{n_0^1} \frac{n_1^1 n_2^1}{n_1^0 n_2^0}\,,\quad\exp(\beta_1)=\frac{n_0^1}{n_0^0} \frac{n_2^0}{n_2^1}\,,\quad\exp(\beta_2)=\frac{n_0^1}{n_0^0} \frac{n_1^0}{n_1^1}\,.$$
The first intuition is to think that coefficients $\beta_1$ and $\beta_2$ do depend on the couples $(n_1^0,n_1^1)$ and $(n_2^0,n_2^1)$ respectively, but it is not the case!

\subsection{Stairs}
\label{A4}
With
$$\mathrm{I}_{p+1}=
  \begin{pmatrix}
   1 & 0 & \cdots & 0 \\
   1 & \ddots& \ddots &  \vdots \\
   \vdots & \ddots & \ddots&  0 \\
   1 & \cdots& 1 & 1 \\   
   \end{pmatrix}\,,\quad\hbox{we find} \quad \mathrm{I}_{p+1}^{-1}=
  \begin{pmatrix}
   1& 0 & \cdots & \cdots & 0 \\
   -1& \ddots & \ddots &  & \vdots \\
   0& \ddots & \ddots& \ddots &  \vdots \\
   \vdots& \ddots & \ddots & \ddots&  0 \\
   0& \cdots & 0& -1 & 1 \\   
   \end{pmatrix}$$

and then 

$$\beta_0=\log \left(\frac{n_0^1}{n_0^0}\right) \,,\quad \beta_i= \log \left(\frac{n_i^1}{n_i^0}\frac{n_{i-1}^0}{n_{i-1}^1} \right)\,,\quad i \in \{1,...,p\}\,.$$

\section{Proof of Proposition 3.1.}\label{appB}

The result is proved with a succession of Taylor expansions of degree $1$ or $2$ in $1/s$. We use
$$\tanh\left(\frac{x}{2} \right) = -1 + 2 e^{x} - 2 e^{2x} + o(e^{2x})\,,$$
then with (\ref{normalgeneral}), we have
\begin{equation}
\label{imbal2}
I_1^Tn^1 = I_1^Tn^1 e^{I_1\beta}+I_0^Tn^0 e^{I_0\beta}- \left(  I_1^Tn^1 e^{2I_1\beta}+I_0^Tn^0 e^{2I_0\beta}\right) + o(|n| \, e^{2\beta_0})\,.
\end{equation}
The first equation of this system gives
$$|n^1| = \sum n^1_i e^{(I_1\beta)_i}+\sum n^0_i e^{(I_0\beta)_i}- \left(  \sum n^1_i e^{2(I_1\beta)_i}+\sum n^0_i e^{2(I_0\beta)_i}\right) + o(|n| \, e^{2\beta_0})\,,$$
therefore, using notations introduced in the proposition, 
$$e^{2\beta_0}\left(n_2^1+s\, n_2^0 \right)-e^{\beta_0}\left(n_1^1+s\, n_1^0 \right) + 1 - o(\frac{1}{s}) = 0\,,$$
and
$$e^{\beta_0} = \frac{1}{2}\frac{n_1^1+s\, n_1^0}{n_2^1+s\, n_2^0} - \frac{1}{2}\frac{n_1^1+s\, n_1^0}{n_2^1+s\, n_2^0}\sqrt{1-4\left( 1-o(\frac{1}{s})\right) \frac{n_2^1+s\, n_2^0}{(n_1^1+s\, n_1^0)^2}}\,.$$
The coefficient $s$ is defined as $s = \frac{|n^0|}{|n^1|}$. Now we have to find the Taylor expansion of $e^{\beta_0}$ of degree two in $1/s$ and reinject it in (\ref{imbal2}). We have
$$\frac{n_2^1+s\, n_2^0}{(n_1^1+s\, n_1^0)^2} = \frac{n_2^1+s\, n_2^0}{s^2 (n_1^0)^2}\frac{1}{(\frac{n_1^1}{s n_1^0}+1)^2} = \frac{n_2^1+s\, n_2^0}{s^2 (n_1^0)^2} \left(1 - 2\frac{n_1^1}{s n_1^0} + o(\frac{1}{s})\right)$$
$$= \frac{1}{s}\frac{n_2^0}{(n_1^0)^2} + \frac{1}{s^2}\left(\frac{n_2^1}{(n_1^0)^2}-2\frac{n_1^1n_2^0}{(n_1^0)^3} \right) + o(\frac{1}{s^2})\,,$$
and
$$\sqrt{1-4\left( 1 - o(\frac{1}{s}) \right)\left(\frac{1}{s}\frac{n_2^0}{(n_1^0)^2} + \frac{1}{s^2}\left(\frac{n_2^1}{(n_1^0)^2}-2\frac{n_1^1n_2^0}{(n_1^0)^3}\right)+ o(\frac{1}{s^2}) \right)}$$
$$=\sqrt{1-4\left(\frac{1}{s}\frac{n_2^0}{(n_1^0)^2} + \frac{1}{s^2}\left(\frac{n_2^1}{(n_1^0)^2}-2\frac{n_1^1n_2^0}{(n_1^0)^3}\right)+ o(\frac{1}{s^2}) \right)}$$

$$=1-2\left(\frac{1}{s}\frac{n_2^0}{(n_1^0)^2} + \frac{1}{s^2}\left(\frac{n_2^1}{(n_1^0)^2}-2\frac{n_1^1n_2^0}{(n_1^0)^3} + \frac{1(n_2^0)^2}{(n_1^0)^4}\right) \right)+ o(\frac{1}{s^2})\,.$$
Therefore
$$e^{\beta_0} = \frac{n_1^1+s\, n_1^0}{n_2^1+s\, n_2^0}\left(\frac{1}{s}\frac{n_2^0}{(n_1^0)^2} + \frac{1}{s^2}\left(\frac{n_2^1}{(n_1^0)^2}-2\frac{n_1^1n_2^0}{(n_1^0)^3} + \frac{1(n_2^0)^2}{(n_1^0)^4}\right)+ o(\frac{1}{s^2}) \right)$$

$$=\frac{n_1^1+s\, n_1^0}{s\, n_2^0}\left(1-\frac{n_2^1}{s n_2^0} + o(\frac{1}{s})\right)\left(\frac{1}{s}\frac{n_2^0}{(n_1^0)^2} + \frac{1}{s^2}\left(\frac{n_2^1}{(n_1^0)^2}-2\frac{n_1^1n_2^0}{(n_1^0)^3} + \frac{(n_2^0)^2}{(n_1^0)^4}\right) + o(\frac{1}{s^2}) \right)$$

$$ = \left(\frac{n_1^0}{n_2^0}+\frac{1}{s}\left(\frac{n_1^1}{n_2^0}-\frac{n_2^1n_1^0}{(n_2^0)^2} \right)\right)\left(\frac{1}{s}\frac{n_2^0}{(n_1^0)^2} + \frac{1}{s^2}\left(\frac{n_2^1}{(n_1^0)^2}-2\frac{n_1^1n_2^0}{(n_1^0)^3} + \frac{(n_2^0)^2}{(n_1^0)^4}\right) \right)+ o(\frac{1}{s^2})$$
$$e^{\beta_0} = \frac{1}{s \, n_1^0}+\left(\frac{1}{s\, n_1^0}\right)^2\left(\frac{n_2^0}{n_1^0}-n_1^1\right) + o(\frac{1}{s^2})\,.$$
The system of equations (\ref{imbal2}) without its first equation is
$$\overline{N}^1 = \frac{\mathtt{I}_1^Tn^1}{|n^1|} = e^{\beta_0}\left(\bold{n_1^1}+s\,\bold{n_1^0} \right) - e^{2\beta_0}\left(\bold{n_2^1}+s\,\bold{n_2^0}\right) + o(\frac{1}{s})$$
and we use the previous expression for $e^{\beta_0}$:
$$\overline{N}^1 = \left(\frac{1}{s\, n_1^0}+\left(\frac{1}{s\, n_1^0}\right)^2\left(\frac{n_2^0}{n_1^0}-n_1^1\right) \right)\left(\bold{n_1^1}+s\,\bold{n_1^0} \right) - \left(\frac{1}{s\, n_1^0}\right)^2\left(\bold{n_2^1}+s\,\bold{n_2^0}\right) + o(\frac{1}{s})$$
$$\overline{N}^1 = \frac{\bold{n_1^0}}{n_1^0} + \frac{1}{s}\left( \frac{n_2^0}{(n_1^0)^2}\left[ \frac{\bold{n_1^0}}{n_1^0}-\frac{\bold{n_2^0}}{n_2^0} \right] + \frac{n_1^1}{n_1^0}\left[ \frac{\bold{n_1^1}}{n_1^1}-\frac{\bold{n_1^0}}{n_1^0} \right]\right) + o(\frac{1}{s})$$
or
$$\overline{N}^1 + \frac{1}{s}\left(\frac{\bold{n_2^0}}{(n_1^0)^2}-\frac{\bold{n_1^1}}{n_1^0} \right) = \frac{\bold{n_1^0}}{n_1^0}\left(1+ \frac{1}{s}\left(\frac{n_2^0}{(n_1^0)^2}-\frac{n_1^1}{n_1^0}\right) \right) + o(\frac{1}{s})\,.$$
Then, using again a Taylor expansion,
$$\frac{\bold{n_1^0}}{n_1^0} = \overline{N}^1 + \frac{1}{s}\left(\frac{\bold{n_2^0}-\overline{N}^1n_2^0}{(n_1^0)^2}-\frac{\bold{n_1^1}-\overline{N}^1n_1^1}{n_1^0} \right)+ o(\frac{1}{s})$$
or 
$$\frac{\bold{n_1^0}}{n_1^0} - \overline{N}^1 = \frac{1}{s}\left(\frac{n_2^0}{(n_1^0)^2}\left(\frac{\bold{n_2^0}}{n_2^0}-\overline{N}^1 \right) - \frac{n_1^1}{n_1^0}\left(\frac{\bold{n_1^1}}{n_1^1}-\overline{N}^1 \right)\right)+ o(\frac{1}{s})\,.$$

\newpage
\section{Path simulations}\label{appC}
In the following graphs, the dotted lines are obtained by the exact path corresponding to Theorem \ref{descriptionLasso} for examples 1 and 2, Theorem \ref{thindep} for examples 3 and 4 and the inclusion case for examples 5 and 6. The solid lines are always given by a coordinate descent algorithm for the standard logistic regression. We change the scale for $t$ (by a linear rescaling) in order to have to same \(\max(t)\) (see (\ref{tseq})) for the exact and algorithmic paths.
\begin{figure}[!ht]
\center
\includegraphics[totalheight=0.3\textheight]{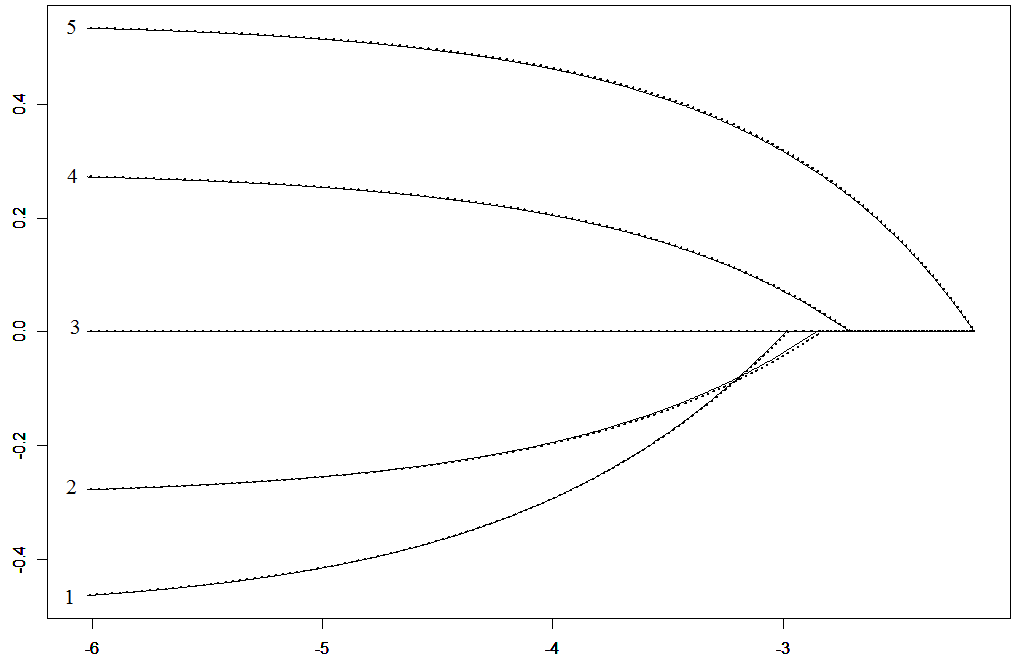} 
\caption{Example 1. Independence and class imbalance. \(\beta_0=-3\) (\(|n^0|/|n^1| = 17\)).}
\label{F1}
\end{figure}

\begin{figure}[!ht]
\center
\includegraphics[totalheight=0.3\textheight]{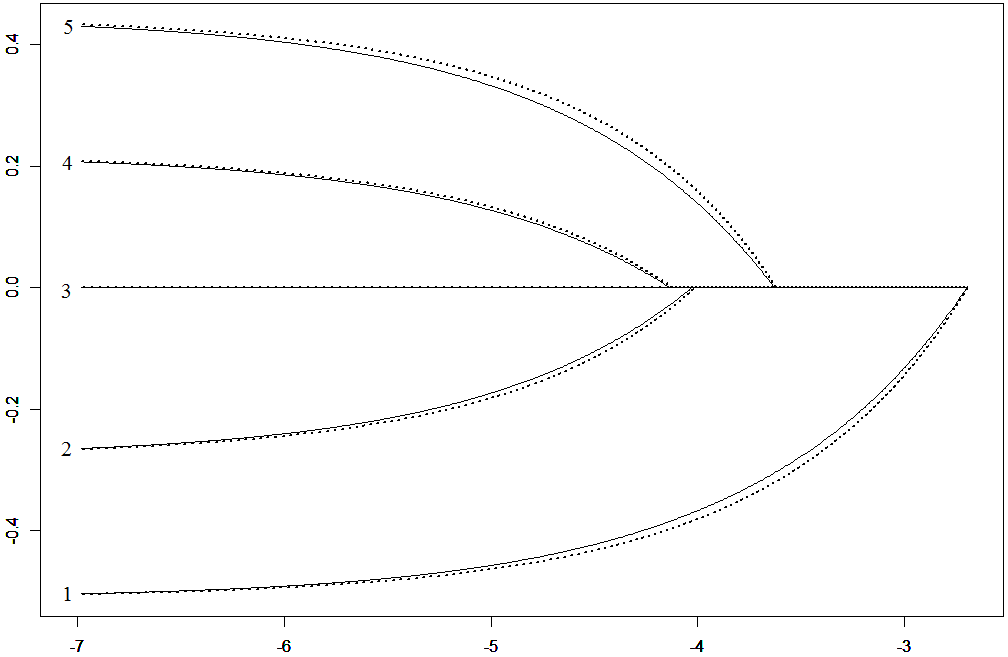} 
\caption{Example 2. Independence and no class imbalance. \(\beta_0 = 0\) (\(|n^0|/|n^1| = 1.5\)).}
\label{F1}
\end{figure}

\newpage

\begin{figure}[!ht]
\center
\includegraphics[totalheight=0.3\textheight]{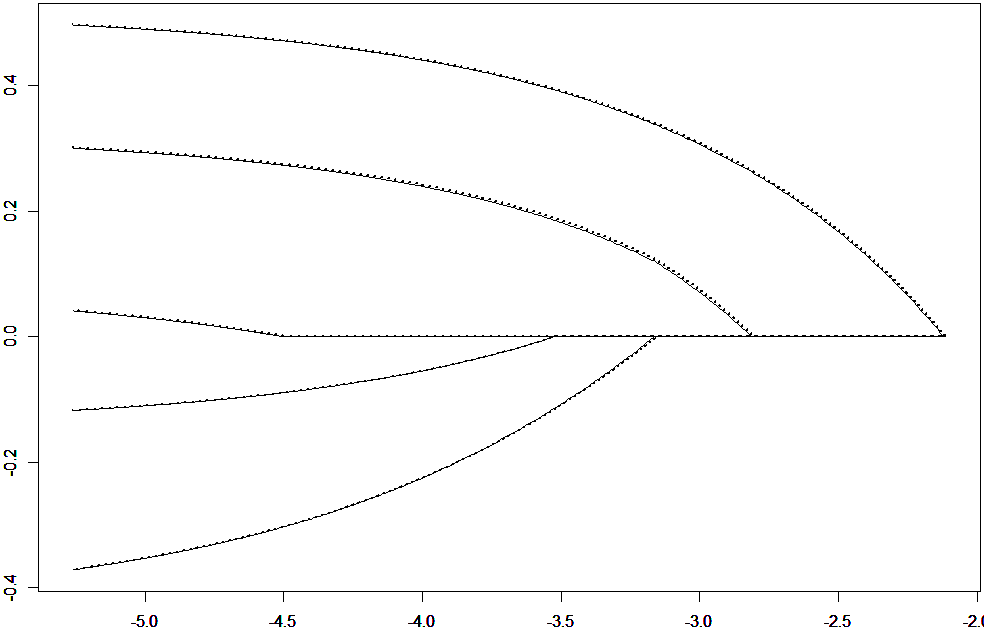} 
\caption{Example 3. Orthogonality and class imbalance. \(\beta_0=-3\) (\(|n^0|/|n^1| = 15\)).}
\label{F1}
\end{figure}

\begin{figure}[!ht]
\center
\includegraphics[totalheight=0.3\textheight]{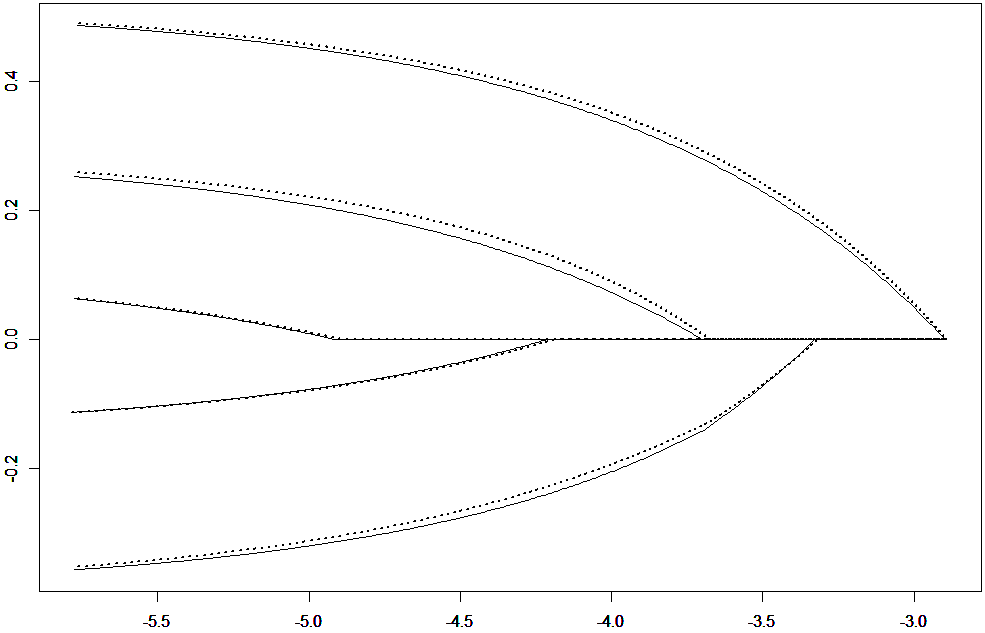} 
\caption{Example 4. Orthogonality and no class imbalance. \(\beta_0=0\) (\(|n^0|/|n^1| = 0.8\)).}
\label{F1}
\end{figure}

\newpage

\begin{figure}[!ht]
\center
\includegraphics[totalheight=0.3\textheight]{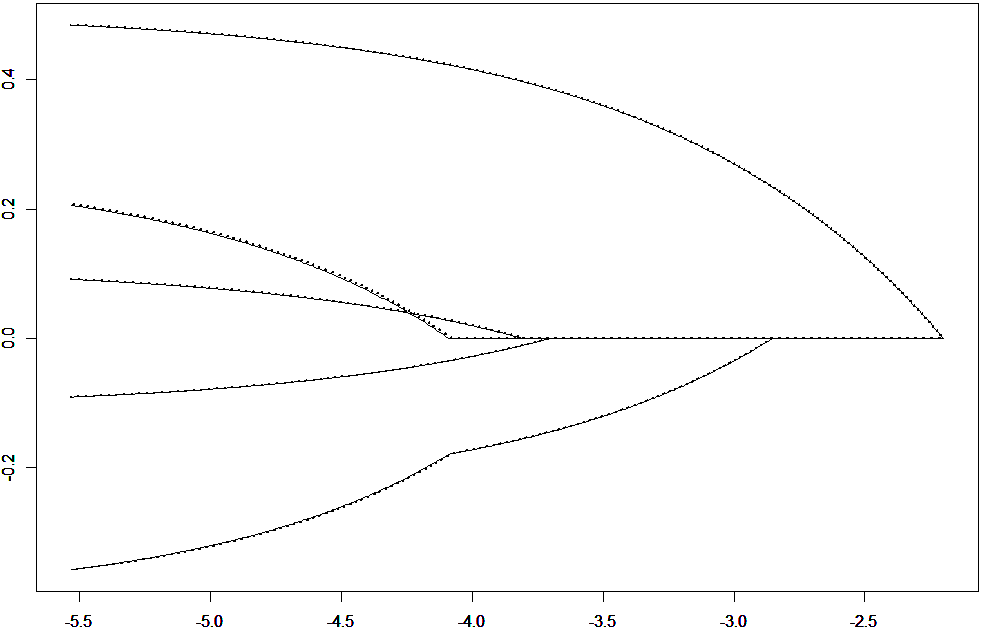} 
\caption{Example 5. Inclusion and class imbalance. \(\beta_0=-3\) (\(|n^0|/|n^1| = 15\)).}
\label{F1}
\end{figure}

\begin{figure}[!ht]
\center
\includegraphics[totalheight=0.3\textheight]{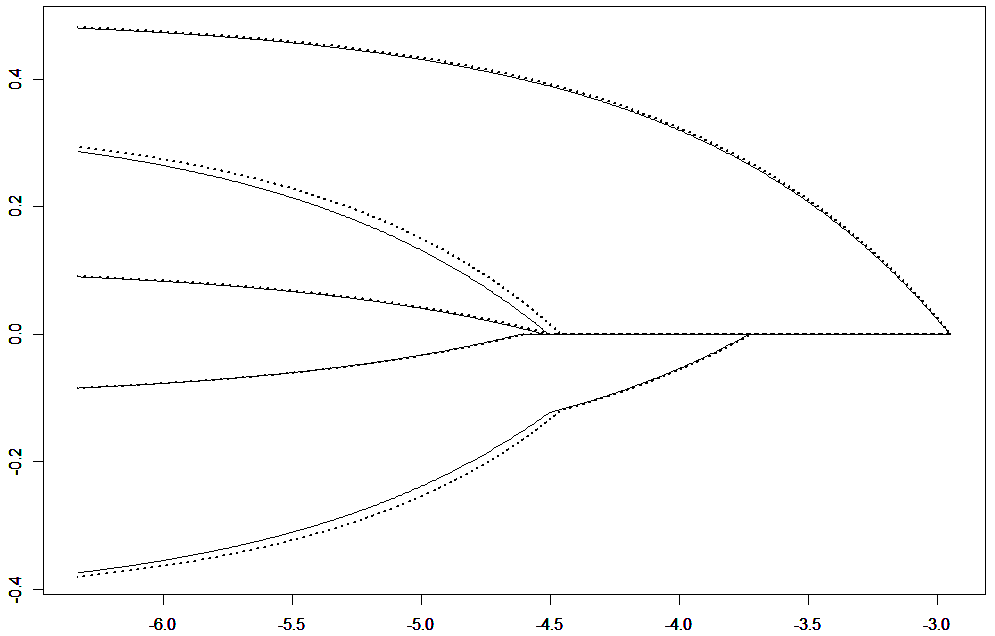} 
\caption{Example 6. Inclusion and no class imbalance. \(\beta_0=0\) (\(|n^0|/|n^1| = 0.8\)).}
\label{F1}
\end{figure}

\end{document}